\documentclass[a4paper,UKenglish,cleveref, autoref, thm-restate]{lipics-v2021}
\usepackage{amsmath, amsthm, amssymb, appendix, bm, graphicx, hyperref, mathrsfs}
\usepackage{multirow}
\usepackage{multicol}
\usepackage{makecell}
\usepackage[ruled]{algorithm2e}
\usepackage{caption}
\usepackage{subcaption}
\title{Dichotomies for \#CSP on graphs that forbid a clique as a minor} 

\author{Boning Meng\footnote{First author.}}{Key Laboratory of System Software (Chinese Academy of Sciences) and State Key Laboratory of Computer Science, Institute of Software, Chinese Academy of Sciences; University of Chinese Academy of Sciences, Beijing 100080, China}{mengbn@ios.ac.cn}{https://orcid.org/0009-0006-0088-1639}{}
\author{Yicheng Pan\footnote{Corresponding author.}}{Beihang University, Beijing 100191, China}{yichengp@buaa.edu.cn}{}{}
\authorrunning{Boning Meng and Yicheng Pan} 

\Copyright{Anonymous} 


\funding{Boning Meng is supported by National Key R\&D Program of China (2023YFA1009500), NSFC 61932002 and NSFC 62272448.}

\nolinenumbers 

\hideLIPIcs  

\ArticleNo{}

\begin{document}

\maketitle

\begin{abstract}
We prove complexity dichotomies for \#CSP problems (not necessarily symmetric) with Boolean domain and complex range on several typical minor-closed graph classes. These dichotomies give a complete characterization of the complexity of \#CSP on graph classes that forbid a complete graph as a minor. In particular, we also demonstrate that, whether the maximum degree of vertices is bounded may influence the complexity on specific minor-closed graph classes, and this phenomenon has never been observed in the previous related studies. Furthermore, our proofs integrate the properties of each graph class with the techniques from counting complexity, and develop a systematic approach for analyzing the complexity of \#CSP on these graph classes.
\end{abstract}

\section{Introduction}

       In this article, we study the counting constraint satisfaction problem (denoted as $\text{\#CSP}$)  with Boolean domain and complex range 
    on several typical minor-closed graph classes\footnote{In this article, we always restrict ourselves to counting problems with Boolean domain and complex range, which means that each variable can only take value in $\{0,1\}$ and each signature has a range over $\mathbb{C}$.}. $\text{\#CSP}$ is considered as one of the most typical framework in counting complexity, as it is capable of expressing a substantial number of counting problems, such as counting the number of vertex covers in a given graph (\#VC) and counting the number of solutions of a CNF formula (\#SAT). In 2014, Cai, Lu and Xia gave a complete complexity dichotomy for $\text{\#CSP}$ on general graphs \cite{cai2014complexity}, and in 2017 Cai and Fu gave a dichotomy for $\text{\#CSP}$ on planar graphs \cite{cai2017holographicuni}. In the latter work, a new tractable case emerges on planar graphs, which is closely related to the problem of counting perfect matchings. 
       
       Counting perfect matchings (denoted as $\text{\#PM}$) is the first natural counting problem discovered to be $\text{\#P}$-hard on general graphs and polynomial-time computable on planar graphs. The complexity of $\text{\#PM}$ has also been studied on several minor-closed graph classes. In 2022, Thilikos and Wiederrecht have presented a complete dichotomy for $\text{\#PM}$ on minor-closed graph classes \cite{thilikos2022killing}, which inspires us to give dichotomies for $\text{\#CSP}$ on these typical graph classes as well. 
        
        This article presents either a polynomial-time algorithm or $\text{\#P}$-hardness for $\text{\#CSP}$ problems on different graph classes. Combining these results together, we give a dichotomy for both $\text{\#CSP}$ and bounded degree $\text{\#CSP}$ over different function sets on graph classes that forbid a complete graph as a minor.

        \subsection{The study of \#PM}\label{PMstudy}
        We start with the introduction of the study of counting perfect matchings, which serves as a background in our study.  
        For a graph $G=(V,E)$, a \textit{matching} is an edge set $M\subseteq E$ such that no pair of edges in $M$ shares a common vertex.  Besides, if the vertices that $M$ contains are exactly $V(M)=V$, then $M$ is denoted as a \textit{perfect matching} of $G$.
        
        An instance of \#PM is a graph $G=(V,E)$ with weighted edges $w:E\to \mathbb{C}$. The weight of a matching $M$ is $w(M)=\prod_{e\in M}w(e)$. The output of the instance is the sum of the weights of all the perfect matchings in $G$:
        $$\text{\#PM}(G)=\sum_{M: M\text{ is a perfect matching of }G}w(M)$$
        When $w(e)=1$ for each $e\in E$, the output of the instance is just the number of perfect matchings in the graph, and we denote this kind of problems as standard \#PM.
            
            \#PM is motivated by the dimer problem in statistical physics \cite{kasteleyn1967graph,kasteleyn1961statistics,kasteleyn1963dimer,temperley1961dimer}, and two fundamental results emerge from this study. The first breakthrough occurred in 1961, when a polynomial time algorithm for $\text{\#PM}$ on planar graphs was developed by Kasteleyn, Temperley and Fisher \cite{kasteleyn1961statistics,temperley1961dimer}, known as the FKT algorithm. The second significant advancement occurred in 1979, when Valiant defined the complexity class $\text{\#P}$ and proved that $\text{\#PM}$ on general graphs is $\text{\#P}$-hard \cite{valiant1979complexity}. 

        The study of $\text{\#PM}$ has attracted considerable interest from two distinct perspectives. One is concerned with the complexity of $\text{\#PM}$ on different graph classes, whereas the other is focused on  the investigation of counting problems that are $\text{\#P}$-hard in general but polynomial-time computable on planar graphs.
        
        Significant progress has been made in both perspectives, and we introduce them in Sections \ref{PMintro} and \ref{countingintro} respectively. In Section \ref{mainresult}, we present our main theorem, which can be seen as an extension to both perspectives. We present the organization of this article briefly in Section \ref{organization}. 
        
\subsection{\#PM on different minor-closed graph classes}\label{PMintro}
  
        To better introduce our study, we begin with some necessary definitions. \textit{Contracting an edge $(u,v)$ in} $G$ means replacing $u,v$ with a new vertex $w$ which is adjacent to all the neighbours of $u,v$. A graph $H$ is a \textit{minor} of $G$ if $H$ can be obtained from $G$ by repeatedly deleting vertices, deleting edges or contracting edges.
        
        A graph class is a (possibly infinite) set of graphs. 
        The graph class of all planar graphs is denoted by $\mathcal{PL}$. 
        An apex graph is defined as a graph that can be embedded on the plane after the removal of a single vertex, which is referred to as the apex vertex of the graph.
        The graph class of all the apex graphs is denoted by $\mathcal{PLA}$.
        
        A graph class $\mathcal{C}$ is said to be \textit{minor-closed} if it is closed under taking minors. If $\mathcal{C}$ is the graph class that forbids a finite set $\mathcal{G}$ of graphs as minors, then $\mathcal{G}$ is the \textit{forbidden minor set} of $\mathcal{C}$, denoted as $\mathcal{C}=fb(\mathcal{G})$.        
        The following theorem shows that the forbidden minor set is a very effective way to characterize graph classes.
     \begin{theorem}[\cite{robertson2004graph}]
     If $\mathcal{C}$ is a minor-closed graph class, then it has a forbidden minor set. 
     \label{thmminor}
     \end{theorem}
        
        We also introduce some notations. We denote $\text{\#PM}$ on the graph class $\mathcal{C}$ as $\text{\#PM}\langle \mathcal{C}\rangle$, and if $\mathcal{C}=fb(\mathcal{G})$, we also denote $\text{\#PM}\langle \mathcal{C}\rangle$ as $\text{\#PM}[\mathcal{G}]$. Further, if $G$ is the only element in $\mathcal{G}$, we denote the problem as $\text{\#PM}[G]$ for convenience\footnote{These notations also applies to other counting problems.}. Using these notations, the results presented in Section \ref{PMstudy} from \cite{kasteleyn1961statistics,temperley1961dimer} and \cite{valiant1979complexity} can be restated as follows. 
            $\text{\#PM}\langle \mathcal{PL}\rangle$, or $\text{\#PM}[\{K_5,K_{3,3}\}]$ can be computed in polynomial time; $\text{\#PM}$, or $\text{\#PM}[\emptyset]$ is \#P-hard.
        
        Several polynomial-time algorithms on other minor-closed graph classes were developed after the FKT algorithm. A polynomial-time algorithm for $\text{\#PM}$ on graphs that can be embedded on a surface of constant Euler genus was developed in 2000 by Galluccio, Loebl and Tesler \cite{galluccio1999theory,tesler2000matchings}, and in 2015 another algorithm without the use of orientations was given by Curticapean and Xia \cite{curticapean2015parameterizing}. In 1989, $\text{\#PM}[K_{3,3}]$ was proved to be computable in polynomial time by Vazirani \cite{vazirani1989nc}. In 2014, a polynomial-time algorithm was developed for $\text{\#PM}[K_5]$ by Straub, Thierauf and Wagner \cite{straub2016counting}, and independently $\text{\#PM}[G]$, where $G$ is a single crossing graph, was proved to be computable in polynomial time as well by Curticapean \cite{curticapean2014counting}. Later in 2019, an NC algorithm for $\text{\#PM}[G]$ was developed by Eppstein and Vazirani for any single crossing graph $G$ \cite{eppstein2019nc}.
        
        Conversely, $\text{\#PM}[K_8]$ is proved to be $\text{\#P}$-hard by Curticapean and Xia in 2022 \cite{curticapean2022parameterizing}, which is the first hardness result on graph classes that forbids a fixed minor set.
        Later in 2022, Thilikos and Wiederrecht figured out the necessary and sufficient condition for $\mathcal{G}$ such that $\text{\#PM}[\mathcal{G}]$ is computable in polynomial time assuming $\text{P}\neq\text{\#P}$, by combining their algorithm with the hardness result for $\text{\#PM}[K_8]$ \cite{thilikos2022killing}. This provides a complete dichotomy for $\text{\#PM}$ on minor-closed graph classes. In \cite{thilikos2022killing}, the authors define the concept of shallow vortex grids, and present the dichotomy based on this concept.

        \begin{definition}[{Shallow vortex grid}{\cite{thilikos2022killing}}]
    The shallow vortex grid of order $k$ is the graph $H_k=(V_k,E_k)$ where $V_k=\{(i,j)|1\le i \le k,1\le j \le 2k\}$ and $E_k=E_{k1}\cup E_{k2}\cup E_{k3}$, where $E_{k1}=\{((i,j),(i+1,j))|1\le i \le k-1,1\le j \le 2k\}$, $E_{k2}=\{((i,j),(i,j+1))|1\le i \le k,1\le j \le 2k-1\}\cup\{((i,2k),(i,1))|1\le i \le k\}$ and $E_{k3}=\{((1,j),(1,j+2))|1\le j \le 2k-2\}\cup\{((1,2k),(1,2)),((1,2k-1),(1,1))\}$.  
    
    The graph class $\mathcal{H}$ is defined as $\{H|H\text{ is a minor of } H_k \text{ for some }k\in \mathbb{N}^+\}$.
    \label{svg}
\end{definition}
       \begin{figure}
            \centering
            \includegraphics[height=0.2\textheight]{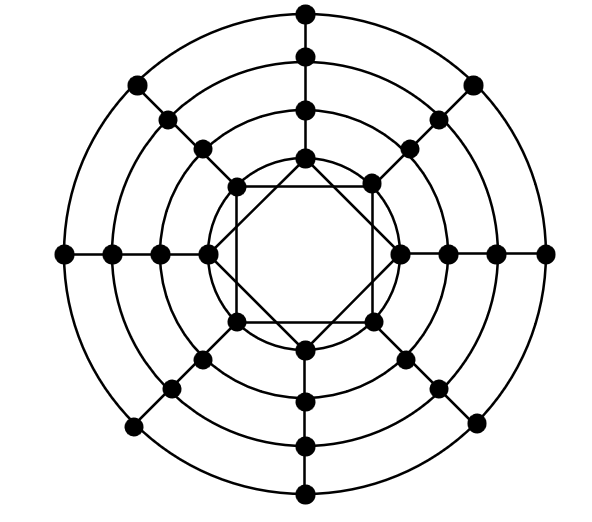}
            \caption{The shallow vortex grid of order 4, denoted as $H_4$}
            \label{fig:H4}
        \end{figure}
        
        See Figure \ref{fig:H4} for an example of shallow vortex grid. The dichotomy is stated as follows.
            \begin{theorem}[\cite{thilikos2022killing}]
            Let $\mathcal{G}$ be a finite set of graphs. If $\mathcal{G}\cap\mathcal{H}\neq \emptyset$, then  $\#PM[\mathcal{G}]$ can be solved in polynomial time; otherwise it is $\text{\#P}$-hard.
            \label{PMalg}
            \end{theorem}
        
\subsection{Different counting problems on planar graphs}\label{countingintro}
  In several frameworks for counting problems, a number of cases that are $\text{\#P}$-hard in general have been found to be polynomial-time computable on planar graphs after \#PM \cite{vertigan2005computational,cai2017holographic,cai2012spin,cai2012gadgets,guo2020complexity,cai2017holographicuni,cai2021new,cai2022fkt,cai2023complexity,cai2023planar}. Most of these cases are related to the FKT algorithm, whereas some of them have been proved not able to be subsumed by the FKT algorithm \cite{cai2021new,cai2022fkt}.
  
  Two frameworks of counting problems, namely $\text{\#CSP}$ problem and $\text{Holant}$ problem, are of great significance in the field of counting complexity, as they are capable of expressing a wide range of problems. For example, problems in \cite{vertigan2005computational,cai2017holographic,cai2012spin,cai2012gadgets} can be expressed in the form of a $\text{\#CSP}$ problem and characterized by the result in \cite{guo2020complexity}. Furthermore, $\text{\#CSP}$ problem and $\text{\#PM}$ problem can be expressed in the form of $\text{Holant}$ problem. We present an informal definition of $\text{\#CSP}$ and $\text{Holant}$ in this section. See Definition \ref{defCSP} and \ref{defHol} in Section \ref{preCSP} for details. 
  \begin{definition}[$\text{\#CSP}$ and $\text{Holant}$]
      $\text{\#CSP}$ is defined by a signature (function) set $\mathcal{F}$, denoted as $\text{\#CSP}(\mathcal{F})$. An instance (input) of $\text{\#CSP}(\mathcal{F})$ is a bipartite graph $G=(U,V,E)$, with each $u\in U$ assigned a signature from $\mathcal{F}$ and each $v\in V$ representing a Boolean variable. The output is the sum of the product of the values of all the signatures, calculated over all possible assignments to the variables.

      The definition of $\text{Holant}$ is identical to that of $\text{\#CSP}$, with the exception that the degree of each $v\in V$ must be exactly 2.
  \end{definition}
     Two key aspects of the complexity of $\text{\#CSP}$ and $\text{Holant}$ problems are often considered in their study. 
  \begin{enumerate}
  	\item The properties of the signature set $\mathcal{F}$. The study of \#CSP and Holant problems typically assumes that the signature set $\mathcal{F}$ is finite. Usually, it also commences with symmetric $\mathcal{F}$, which means that the value of each signature in $\mathcal{F}$ only depends on the number of $1$s in the input. This kind of problem is denoted as sym-$\text{\#CSP}$ or sym-$\text{Holant}$. It is also sometimes beneficial to consider that $\mathcal{F}$ contains certain kinds of signatures \cite{cai2009holant,cai2011dichotomy,cai2018dichotomy,backens2017complete}, which is not the focus of this article.
  	\item The graph class that the underlying graph $G$ is restricted to. In previous research, \#CSP and Holant problems have been studied on general graphs and planar graphs. This article focuses on these problems on minor-free graphs. In addition, whether the maximum degree of $G$ is bounded also plays an important role in this study.
  \end{enumerate}
    
    Several dichotomies have been demonstrated on general graphs and planar graphs in the study of $\text{\#CSP}$ and $\text{Holant}$ problems. Please refer to Table \ref{dicbef} for corresponding references and the position of our results.

\begin{table}[]
\centering
\begin{tabular}{|l|l|l|l|}
\hline
           & General graphs         & Planar graphs & Minor free graphs \\ \hline
\#PM       & \cite{valiant1979complexity} & \cite{kasteleyn1961statistics,temperley1961dimer} &\cite{thilikos2022killing,curticapean2022parameterizing}                   \\ \hline
sym-\#CSP  & \multirow{2}{*}{\cite{cai2014complexity}} & \cite{guo2020complexity}& \multirow{2}{*}{\textbf{Our results}}          \\ \cline{1-1} \cline{3-3} 
\#CSP      &                       & \cite{cai2017holographicuni}&                  \\ \hline
sym-Holant & \cite{cai2013complete} & \cite{cai2022fkt}         &      open             \\ \hline
Holant     &      open            &    open    &     open       \\ \hline
\end{tabular}
\caption{References related to \#PM and dichotomies of \#CSP and Holant problems on sorts of graphs. Our results fill the blank of $\text{\#CSP}$ on minor-free graphs.}
\label{dicbef}
\end{table}

  \subsection{Our results and approaches}\label{mainresult}
   
Our study settles the $\text{\#CSP}$ problem on several minor-closed graph classes that have been identified as typical and pivotal scenarios in the previous research on \#PM \cite{curticapean2014counting,curticapean2022parameterizing,thilikos2022killing}. The main findings of this study include several dichotomies on these minor-closed graph classes. This is the first counting complexity dichotomy to consider both signature set and graph class at the same time, thereby bridging the gap between the study of complexity classification on signature sets and that on graph classes. Here, due to the lack of introduction of several key concepts, we present a weak but more comprehensible version of our findings (Theorem \ref{allthm}) first, and then a strong version (Theorem \ref{allthm2}, Section \ref{sec:preprocessing}) that we actually prove after all necessary definitions have been provided.
   
\begin{theorem}
 For $\text{\#CSP}$ on Boolean domain with complex range:
      \begin{enumerate}
          \item $\text{\#CSP}(\mathcal{F})[K_4]$ can be computed in linear time for arbitrary finite $\mathcal{F}$; 
          \item $\text{\#CSP}(\mathcal{F})[K_5]$ is computable in polynomial time if $\text{\#CSP}(\mathcal{F})\langle \mathcal{PL}\rangle$ is polynomial-time computable; otherwise it is $\text{\#P}$-hard. 
          \item $\text{\#CSP}(\mathcal{F})[K_6]$ is computable in polynomial time if $\text{\#CSP}(\mathcal{F})$ is polynomial-time computable; otherwise it is $\text{\#P}$-hard. 
          \item For any constant $k\in \mathbb{N}^+$, maximum degree $k$ $\text{\#CSP}(\mathcal{F})[ K_7 ]$ is computable in polynomial time if $\text{\#CSP}(\mathcal{F})\langle \mathcal{PL}\rangle$ is polynomial-time computable; otherwise it is $\text{\#P}$-hard. 
          \item For for any integer $k\ge 3$, maximum degree $k$ $\text{\#CSP}(\mathcal{F})[ K_8 ]$ is computable in polynomial time if $\text{\#CSP}(\mathcal{F})$ is polynomial-time computable; otherwise it is $\text{\#P}$-hard. 
      \end{enumerate}
      \label{allthm}
  \end{theorem}

\begin{table}[]
\centering
\begin{tabular}{|l|llllll|}
\hline
    & \multicolumn{1}{l|}{$K_4$}                  & \multicolumn{1}{l|}{$K_5,K_{3,3}$} & \multicolumn{1}{l|}{$K_5$} & \multicolumn{1}{l|}{$K_6$} & \multicolumn{1}{l|}{$K_7$} & $K_8$     \\ \hline
$\mathcal{F}\subseteq \mathscr{A}$                                                                & \multicolumn{6}{l|}{\multirow{2}{*}{$\text{P}$ \cite{cai2014complexity}}}                                                                                                                                         \\ \cline{1-1}
$\mathcal{F}\subseteq \mathscr{P}$                                                                & \multicolumn{6}{l|}{}                                                                                                                                                               \\ \hline
$\mathcal{F}\subseteq \widehat{\mathscr{M}_P}$, and $\mathcal{F}\nsubseteq \mathscr{A},\mathscr{P}$ & \multicolumn{1}{l|}{\multirow{2}{*}{$\text{P}$ (1)}} & \multicolumn{1}{l|}{$\text{P}$ \cite{cai2017holographicuni,meng2025matchgate}}         & \multicolumn{1}{l|}{$\text{P}$ (2)} & \multicolumn{2}{l|}{\#P,P (3),(4)}                            & $\text{\#P}$ (5) \\ \cline{1-1} \cline{3-7} 
Otherwise                                                                                         & \multicolumn{1}{l|}{}                       & \multicolumn{5}{l|}{$\text{\#P}$ \cite{cai2017holographicuni,meng2025matchgate}}                                                                                                        \\ \hline
\end{tabular}
\caption{A summary of Theorem  \ref{allthm}. Each row denotes a certain case for $\mathcal{F}$. Each column denotes the forbidden minors of the graph class that the underlying graph is restricted to. The symbol ``$\text{P}$'' denotes the corresponding problem is polynomial-time computable, while ``$\text{\#P}$'' denotes the corresponding problem is $\text{\#P}$-hard.  In particular, ``\#P,P'' denotes the corresponding problem is $\text{\#P}$-hard in general, but has a polynomial-time algorithm when the degree of each vertex is upper bounded. We use (1)-(5) to indicate the conclusions corresponding to the five statements in Theorem  \ref{allthm}.}
\label{alltab}
\end{table}

Please refer to Section \ref{sec:precspdicho} and Section \ref{sec:preprocessing} for the detailed condition of being polynomial-time computable.
We summarize the results of Theorem \ref{allthm} in Table \ref{alltab}. Please refer to Definitions \ref{defA}, \ref{defP}, \ref{def:MP} and Theorem \ref{thmHT} for corresponding definitions of $\mathscr{A},\mathscr{P}$ and $\widehat{\mathscr{M}_P}$. We also remark that each statement in Theorem \ref{allthm} can be viewed as a corollary of the corresponding statement in Theorem \ref{allthm2}, which is the complete version of our results. 
    
    It is a bit surprising that an unbounded maximum degree of vertices will turn over the complexity of \#CSP problem on certain graph classes. This phenomenon has never been observed in the previous study of \#PM or \#CSP problem. Here, we explain this phenomenon informally. When the maximum degree is not bounded, we may not enumerate all the valid assignments of edges incident to a single vertex, as the number of these assignments can be exponential in $|G|$.
    Consequently, the deletion of a single vertex from $G$, as performed in the algorithm described in "Killing a vortex" \cite{thilikos2022killing}, is no longer an available approach, since it would result in an exponential increase in the time complexity. 
    This means that for \#CSP problem in which the maximum degree is unbounded, after ``killing a vortex'', we must further ``kill an apex'' (Definition \ref{ga3}), to obtain a polynomial-time algorithm.

    We prove Theorem \ref{allthm} (and Theorem \ref{allthm2}) in Sections \ref{secK5}-\ref{secK6}. Basically, the proof consists of three parts. The first part includes the first statement in Theorem \ref{allthm} (and also in Theorem \ref{allthm2}), which corresponds to the column ``$K_4$'' in Table \ref{alltab}. In fact, this statement can be straightforwardly proved by the result in \cite{ganian2018sum}, which is stated as Theorem \ref{incidence}. 
    
    In the second and third parts, with four existing dichotomies for \#CSP, it is sufficient for us to focus on the case when \#CSP$(\mathcal{F})$ is hard in general but polynomial-time computable on planar graphs, which corresponds to the row ``$\mathcal{F}\subseteq \widehat{\mathscr{M}_P}$, and $\mathcal{F}\nsubseteq \mathscr{A},\mathscr{P}$'' in Table \ref{alltab}. In this case, we transform \#CSP$(\mathcal{F})$ into Holant$(\widehat{\mathcal{F}}|\widehat{\mathcal{EQ}})$ under a holographic transformation described in Theorem \ref{thmHT}. After the transformation, each signature in $\widehat{\mathcal{F}}\cup \widehat{\mathcal{EQ}}$ can be expressed by a matchgate, and we have $\widehat{\mathcal{F}}\nsubseteq \mathscr{A}$. 
    
    We also remark that Theorem \ref{allthm} (and Theorem \ref{allthm2}) still holds even when signatures in $\mathcal{F}$ are not restricted to be symmetric. To achieve this, we refer to the results in \cite{meng2025matchgate}, in which the properties of permutable matchgate signatures are characterized in detail. Using these properties, we are able to transform a \#CSP problem into a sym-\#CSP problem in the algorithm part, and reduce a sym-\#CSP problem to the the \#CSP problem in the hardness reductions. 
    
  The second part includes the second and the fourth statements of Theorem \ref{allthm} (and also of Theorem \ref{allthm2}). In this part, we extend two of the algorithms for \#PM \cite{curticapean2014counting,thilikos2022killing} to an algorithm for Holant$(\mathcal{F})$ and another one for bounded degree Holant$(\mathcal{F})$ on certain graphs, where each signature in $\mathcal{F}$ is a permutable matchgate signature. In each algorithm, for a given graph $G$ with each vertex assigned a signature from $\mathcal{F}$, we construct a specific tree decomposition $(T,\beta)$ of $G$, and set a node in the tree decomposition as the root. This step derives from $G$ a tree-like structure that consists of $|T|$ bags, and we seek to compute the value of $G$ from the leaves to the root recursively. 
  
  It is noteworthy that in the two algorithms in \cite{curticapean2014counting,thilikos2022killing}, a necessary condition for this recursive computation is that the size of each intersection of two adjacent bags is at most 3. Consequently, if we replace each vertex with a 
 corresponding matchgate naively, this size requirement would be destroyed. Also, if the replacement is performed at the beginning of the algorithm, the minor-forbidden property may be violated. To overcome this difficulty, we introduce the \emph{path gadget} to replace each vertex appearing in multiple bags, and appropriately modify the tree decomposition. This allows us to compute the corresponding value of a single bag in the graph without violating the size requirement. Then, we record the corresponding value of each bag as either a \emph{representative signature} or a \emph{boundary mapping}, and transform the effect of child bags into a single introduced vertex in the parent bag. We also demonstrate that this transformation would not change the value of $G$ as well as the properties that hold for each bag. Finally, we prove that these algorithms can output the correct value by recursions.

    The third part contains the third and the fifth statements. In this part, we prove \#P-hardness for Holant$(\mathcal{F}|\widehat{\mathcal{EQ}})$ and bounded degree Holant$(\mathcal{F}|\widehat{\mathcal{EQ}})$ on certain graph classes, which are the bipartite form of Holant problems, to obtain the desired \#P-hardness. Here, $\mathcal{F}\subseteq \mathscr{M}_P$,$\mathcal{F}\nsubseteq \mathscr{A}$, and $\widehat{\mathcal{EQ}}$ is a specific set of signatures which will be defined in Section \ref{preliminareis}. By the characterization of permutable matchgate signatures, there exists $f\in \mathcal{F}$ satisfying $f\in \mathscr{M}_P-\mathscr{A}$, and therefore a symmetric $g\in \mathscr{M}_P-\mathscr{A}$ can be realized.
    
    Then we start our reduction with the fact that, counting perfect matchings on 3-regular graphs and counting matchings on planar 3-regular graphs are \#P-hard. We reduce each problem to a corresponding intermediate problem in which the underlying graph comes into the graph class that Holant$(\mathcal{F}|\widehat{\mathcal{EQ}})$ or bounded degree Holant$(\mathcal{F}|\widehat{\mathcal{EQ}})$ is also restricted to. We then simulate each signature in the intermediate problem with signatures in $\{g\}\cup \widehat{\mathcal{EQ}}$ in a planar and bipartite manner by gadget construction and polynomial interpolation. After the simulation, we demonstrate that the underlying graph of each obtained instance of Holant$(\mathcal{F}|\widehat{\mathcal{EQ}})$ or bounded degree Holant$(\mathcal{F}|\widehat{\mathcal{EQ}})$ remains in the origin graph classes, which yields the \#P-hardness results.
    
    Besides, it should be noted that the majority of the results presented in the second and third parts are not straightforward. This is due to the necessity of classifying corresponding signatures into multiple categories and analyzing each category individually in the proofs. Furthermore, the construction of the gadgets in the reduction process is challenging. In most reductions, the gadget construction must be done in a bipartite, planar manner, and sometimes with the restriction that the degree of each vertex is bounded by 3. 
    


\subsection{Organization}\label{organization}
This paper is organized in order of dependency. 
In Section \ref{preliminareis}, we introduce the preliminaries needed in our proof. 
In Section \ref{sec:preprocessing}, we present the strong version of our main results. 
We prove the second to the fifth statements in Theorem \ref{allthm} and \ref{allthm2} in Sections \ref{secK5}, \ref{secK6}, \ref{secK7}, \ref{secK8}, respectively.
In Section \ref{conclusion}, we conclude our paper.
\section{Preliminaries}\label{preliminareis}

\subsection{Tree decompositions and graph parameters}

\subsubsection{Tree decompositions}
Tree decompositions are closely related to minor theory and can often be used for algorithm design.
        \begin{definition}[{Tree decomposition}\cite{ROBERTSON1986309}]
            A tree decomposition $(T,\beta)$ of a graph $G=(V,E)$ consists of a tree $T$ and a mapping $\beta:V(T)\to2^{V(G)}$ that maps every node in $T$ to a subset of $V(G)$ such that:
            \begin{itemize}
                \item $\bigcup_{t\in T}\beta(t)=V(G)$
                \item $\forall uv\in E, \exists t\in T$ such that $u,v\in \beta(t)$
                \item $\forall v\in V$, $\{t\in T|v\in \beta(t)\}$ induces a subtree in $T$.
            \end{itemize}
            \label{deftw}
        \end{definition}
        The \textit{width} of a tree decomposition $(T,\beta)$ is $\max_{t\in T} |\beta(t)|-1$, and the \textit{treewidth} of a graph $G=(V,E)$, denoted as $\textsf{tw}(G)$, is the minimum width over all possible tree decompositions of $G$. For each $k\in \mathbb{N}^+$, the graph class of all graphs with treewidth no greater than $k$ is defined to be $\mathcal{TW}_k=\{G|\textsf{tw}(G)\le k\}$. For $t\in T$, the \textit{torso} of $G$ at $t$ is $G_t=(\beta(t),E_t)$, where $E_t=\{uv\in E|u,v\in \beta(t)\}\cup \{uv|u,v\in \beta(t), \exists dt\in T, u,v\in \beta(d)\}$.
        
        A tree decomposition is called a \textit{normal tree decomposition} if for each $dt\in T$, neither $\beta(t)\subseteq \beta(d)$ nor $\beta(d)\subseteq \beta(t)$ holds.
        
        \begin{lemma}
        A normal tree decomposition can be obtained from an arbitrary tree composition $(T,\beta)$ in $O(|T|^2|G|)$ time.
            \label{lemnormaltree}
        \end{lemma}
        
        \begin{proof}
            If $dt\in T$ and $\beta(t)\subseteq \beta(d)$, by contracting $dt$ to $w$ we get another tree $T'$ . We also define a mapping $\beta':V(T')\to 2^{V(G)}$ such that $\beta'(w)=\beta(d)$ and $\beta'(s)=\beta(s)$ when $s\neq w$. It can be verified that $(T',\beta')$ also forms a tree decomposition. Deciding whether there exists such $dt$ requires $O(|T||G|)$ time while contracting $dt$ requires $O(1)$ time. Since the number of edges in the tree is bounded by $|T|$, we can obtain a normal tree decomposition in $O(|T|^2|G|)$ time by recursions. 
        \end{proof}
        
        The size of a normal tree decomposition is bounded:
        
        \begin{lemma}
        For a normal tree decomposition $(T,\beta)$ of $G$, $|T|\le |G|$. 
            \label{lemsizebound}
        \end{lemma}
        \begin{proof}
            For a leaf $l\in T$, there is an edge $lt\in T$. Since $(T,\beta)$ is a normal tree decomposition, $\beta(l)\nsubseteq \beta(t)$ and there is at least one vertex $v_l\in \beta(l)$ such that for each $d\in T-l$, $v_l\notin \beta(d)$. After deleting $l$ from $T$, $(T-l,\beta)$ is also a normal tree decomposition of $G'\subseteq G$ satisfying $V(G')\subseteq V(G)-v_l$. This procedure can be done recursively with at least one vertex in $G$ deleted in each round, which yields $|T|\le |G|$.
        \end{proof}
\subsubsection{Graph parameters}
        
        A \textit{graph parameter} is a mapping which maps a graph to a non-negative integer. For two graph parameters $\textsf{p}$ and $\textsf{p'}$, we define $\textsf{p}\preceq \textsf{p'}$ to mean that there exists a computable function $f:\mathbb{N}\to\mathbb{N}$ such that for each graph $G$, $\textsf{p}(G)\le f(\textsf{p'}(G))$. We also say that $\textsf{p}$ and $\textsf{p'}$ are \textit{asymptotically equivalent} if $\textsf{p}\preceq \textsf{p'}$ and $\textsf{p'}\preceq \textsf{p}$, denoted as $\textsf{p}\sim \textsf{p'}$. 
        
         In the following, we define several graph parameters related to minors and tree decompositions, as presented in reference \cite{thilikos2022killing}. We confine ourselves to a self-contained fragment of their definitions and make some modifications in order to provide a more intelligible explanation. 
        
        The \textit{Hadwiger number} of a graph $G=(V,E)$, denoted as $\textsf{hwg}(G)$, is the maximum $t$ such that $G$ has a $K_t$ minor. Similarly, the maximum $t$ such that $G$ has a $H_t$ minor can be denoted as $\textsf{svg}(G)$. Here, $K_t$ denotes the clique with $t$ vertices while $H_t$ is defined in Definition \ref{svg}. 
        We use $\textsf{genus}(G)$ to denote the minimum Euler-genus of a surface where $G$ can be embedded on without crossing and $\textsf{apex}(G)$ to denote the minimum size of $U\subseteq V$ such that $G-U$ is planar. Furthermore, a series of graph parameters can be defined based on tree decompositions.
        
        \begin{definition}
           For a graph $G=(V,E)$ and a tree decomposition $(T,\beta)$ of $G$, the $ga3$-width of $(T,\beta)$ is the minimum $k$ such that for every $t\in T$, if the size of the torso $|G_t|>k$, then $G_t$ must satisfy the following conditions:
            \begin{itemize}
                \item There is an apex set $A_t\subseteq V(G_t)$, a surface $\Sigma_t$ such that $G_t-A_t$ can be embedded on $\Sigma_t$ without crossing;
                \item (g) $\Sigma_t$ has Euler-genus at most $k$;
                \item (a) $|A_t|\le k$;
                \item (3) For any $dt\in T$, $\beta(d)\cap\beta(t)$ can only avoid $A_t$ with at most 3 vertices, which means $|\beta(d)\cap \beta(t)-A_t|\le 3$. Besides, if $|\beta(d)\cap \beta(t)-A_t|=3$, then $\beta(d)\cap \beta(t)-A_t$ must bound a face in the embedding.
            \end{itemize}
            We use $\textsf{p}_{ga3}(G)$ to denote the minimum $ga3$-width over all possible tree decompositions of $G$.
            \label{ga3}
        \end{definition}
        Definition \ref{ga3} can also define 7 other parameters by independently replacing $g,a,3$ with $-,-,+$. By replacing $g$ and/or $a$ with $-$, we set the $k$ in the corresponding condition to $0$. By replacing $3$ with $+$, we delete the condition $(3)$. To illustrate, $\textsf{p}_{ga+}=k$ signifies the existence of a tree decomposition such that every torso of it can be embedded on a surface with Euler-genus at most $k$ after deleting at most $k$ vertices. In contrast, $\textsf{p}_{--3}=k$ denotes the necessity for a tree decomposition whose torsos with size greater than $k$ are planar and satisfy the (3) condition. 
        
   \begin{remark}
       The definitions presented here differ in two aspects from those in ``Killing a vortex'' \cite{thilikos2022killing}. Firstly, the concept of "vortex" is omitted, as it has already been "killed" by \cite{thilikos2022killing}, which means that it is no longer necessary to consider this concept in this article. Secondly, the (3) condition is introduced into the definitions, as it plays a pivotal role in the algorithms in the following sections. For illustration, the parameter $\textsf{p}_{ga+}$ in this article is equivalent to the parameter "$\textsf{p}_{-ga}$" in \cite{thilikos2022killing}. 
   \end{remark}
   There are some trivial relations between these parameters:
   \begin{lemma}
       For all $i\in\{g,-\},j\in\{a,-\},k\in\{3,+\}$,
       $$\textsf{p}_{ij+}\le \textsf{p}_{ij3},\textsf{p}_{iak}\le \textsf{p}_{i-k},\textsf{p}_{gjk}\le \textsf{p}_{-jk}$$
       \label{lattice}
   \end{lemma}
   A number of significant theorems and algorithms pertaining to these parameters have also been identified:
   \begin{theorem}[\cite{thilikos2022killing}]
       $\textsf{svg} \preceq \textsf{p}_{ga+},\textsf{p}_{ga3}\preceq \textsf{svg}$.
       \label{thmmtotw}
   \end{theorem}
   By combining Lemma \ref{lattice} with Theorem \ref{thmmtotw}, we obtain the following corollary:
   \begin{corollary}
       $\textsf{svg} \sim \textsf{p}_{ga+}\sim \textsf{p}_{ga3}$.
   \end{corollary}
   Furthermore, for a graph $G$, if $\textsf{svg}(G)$ is given, then the corresponding tree decomposition with bounded $ga3$-width can be constructed in polynomial time:
   \begin{theorem}[\cite{thilikos2022killing}]
       For a graph $H$ which is a minor of $H_k$ and a graph $G$, in $O(f(k)\cdot|G|^3)$ time where $f(k)$ is some computable function, we can find either the fact that $H$ is a minor of $G$, or alternatively a tree decomposition $(T,\beta)$ of $G$ whose $ga3$-width is less than $c(k)$, where $c(k)$ is some computable function that depends only on $k$. 
    \label{K7trde}
   \end{theorem}
     A \textit{single crossing graph} is defined as a graph that can be embedded on the plane with at most one crossing. Such graphs exhibit the following property:
     \begin{theorem}[\cite{curticapean2014counting}]
   For a single crossing graph $H$,  there exists a constant $c$ such that for all graphs $G\in fb(\{H\})$, $\textsf{p}_{--3}(G)\le c$. Furthermore, a tree decomposition with $--3$-width less than $c$ can be found in $O(|G|^4)$. 
   \label{K5trde}
   \end{theorem}

\subsection{Counting problems}\label{preCSP}
      For a string $\alpha=\alpha_1\dots\alpha_k \in \{0,1\}^k$, the \textit{Hamming weight} of $\alpha$ is the number of $1$s in $\alpha$, denoted as $HW(\alpha)$. We use $\overline{\alpha}$ to denote the string that differs from $\alpha$ at every bit, which means $\alpha_i+\overline{\alpha}_i=1$ for each $1\le i\le k$. A \textit{signature}, or a \textit{constraint function}, is a function $f:\{0,1\}^k \to\mathbb{C}$, where $k$ is the \textit{arity} of $f$.  A signature $f$ is said to be \textit{symmetric} if the value of $f$ depends only on the Hamming weight of the input string. A symmetric signature $f$ of arity $k$ can be denoted as $[f_0,f_1,...,f_k]_k$, or simply $[f_0,f_1,...,f_k]$ when $k$ is clear from the context. Here, for $0\le i\le k$, $f_i$ is the value of $f$ when the Hamming weight of the input string is $i$. For $c\in \mathbb{C}$, we also use the notation $c[f_0,f_1,...,f_k]$ to denote the signature $[cf_0,cf_1,...,cf_k]$. We use $\leq_T$ and $\equiv_T$ to respectively denote polynomial-time Turing reduction and equivalence. 
     
     We denote by $f^{x_i=c}$ the signature that pins the $i$th variable to $c\in\{0,1\}$:
     
     $$f^{x_i=c}(x_1,...,x_{i-1},x_{i+1},...,x_{k})=f(x_1,...,x_{i-1},c,x_{i+1},...,x_{k})$$
     
     For a string $\alpha\in \{0,1\}^q, q\le k$, we also define $f^\alpha=f^{x_1=\alpha_1,x_2=\alpha_2,...,x_q=\alpha_q}$, where $\alpha_i$ is the $i$th bit of $\alpha$ for $1\le i\le q$. That is, $f^\alpha$ is obtained from $f$ by pinning the first $q$ bits of $f$ to $\alpha_1,...,\alpha_q$. Equivalently speaking, $\alpha$ can be seen as an assignment to the $q$ variables $\alpha:\{1,\dots,q\}\to \{0,1\}$. In the case of a slight overuse, the notation $f^\alpha$ is sometimes employed to indicate the signature by pinning specific $q$ bits of $f$ to $\alpha_1,...,\alpha_q$ when the $q$ bits are clear from the context. 
     
     In this article, we mainly focus on two frameworks of counting problems: the counting constraint satisfaction problem and the Holant problem.
     
     \subsubsection{\#CSP problems}
       A \textit{counting constraint satisfaction problem} $\text{\#CSP}(\mathcal{F})$ \cite{creignou2001complexity} requires the value of an instance, which is the sum of the values over all configurations. Here, $\mathcal{F}$ is a fixed and finite set of signatures. An instance of $\text{\#CSP}(\mathcal{F})$ is specified as follows:
      
      \begin{definition}
            An \textit{instance $I$ of $\text{\#CSP}(\mathcal{F})$} has $n$ variables and $m$ signatures from $\mathcal{F}$ depending on these variables. The value of the instance then can be written as
            
        $$Z(I)=\sum_{(x_1,...,x_n)\in \{0,1\}^n}\prod_{1\le i\le m} f_i(x_{i_1},...,x_{i_k})$$
        where $f_1,\dots,f_m$ are signatures in $I$ and $f_i$ depends on $x_{i_1},...,x_{i_k}$ for each $1\le i\le m$.
          \label{defCSP}
      \end{definition}
        The \textit{underlying graph} of a $\text{\#CSP}(\mathcal{F})$ instance $I$ is a bipartite graph $G=(U,V,E)$, where for every constraint $f$ there is a $u_f\in U$, for every variable $x$ there is a $v_x\in V$, and $(u_f,v_x)\in E$ if and only if $f$ depends on $x$. We remark that the treewidth of $G$ is also denoted as the  incidence treewidth of $I$, and in this article when we refer to treewidth, we always mean incidence treewidth\footnote{In contrast,  a graph $G'=(V,E)$  can also be obtained from a $\text{\#CSP}(\mathcal{F})$ instance $I$, where for every variable $x$ there is a $v_x\in V$, and $(v_x,v_y)\in E$ if and only if they appear together in the scope of a constraint. The treewidth of such $G'$ is denoted as the primal treewidth of $I$, which is of independent interest and will not be used in this article.}. Sometimes we also denote the value $Z(I)$ as $Z(G)$ for convenience. If each constraint function in $\mathcal{F}$ is restricted to be symmetric, we denote this kind of problem as symmetric \#CSP, or sym-\#CSP for short. 
        
        If we restrict the maximum degree of the vertices in $U\cup V$ to be no more than a constant $k$, we denote this kind of problem as \textit{bounded degree} \#CSP, or \textit{maximum degree} $k$ \#CSP.  If we restrict the maximum degree of the vertices in $V$ to be no more than a constant $D$, this kind of problem is denoted as $\#R_D$-CSP in \cite{cai2014complexity}.
        As $\mathcal{F}$ is finite, the maximum arity of signatures in $\mathcal{F}$ is also finite, and we denote this integer as $k$. Consequently, the maximum degree of the vertices in $U$ is no more than $k$ as well. As result, by these definitions we have the following lemma.
        \begin{lemma}
        Suppose the maximum arity of signatures in $\mathcal{F}$ is $k$. For each integer $D\ge1$,
        
        $$\text{maximum degree } D \text{ \#CSP}(\mathcal{F})\le_T \#R_D\text{-CSP}(\mathcal{F})$$
        $$\#R_D\text{-CSP}(\mathcal{F}) \le_T \text{maximum degree } \max{(k,D)} \text{ \#CSP}(\mathcal{F})$$
            \label{RD=bounded}
        \end{lemma}
        
        \subsubsection{Holant problems}
        A \textit{Holant problem} $\text{Holant}(\mathcal{F})$ can be seen as a $\text{\#CSP}(\mathcal{F})$ problem with the restriction that all the variables must appear exactly twice. 
        
        \begin{definition}
        An instance of $\text{Holant}(\mathcal{F})$ has an underlying graph $G=(V,E)$. Each vertex $v\in V$ is assigned a signature from $\mathcal{F}$ and each edge in $E$ represents a variable.  Here, $\mathcal{F}$ is a fixed set of signatures and usually finite. The signature assigned to the vertex $v$ is denoted as $f_v$. An assignment of $E$ is a mapping $\sigma:E\to \{0,1\}$, which can also be expressed as an assignment string $\sigma\in \{0,1\}^{|E|}$, and the value of the assignment is defined as
         $$\omega(\sigma)=\prod_{v\in V} f_v(\sigma)$$
        where $f_v(\sigma)=f_v(\sigma({e_{v_1}}),...,\sigma({e_{v_k}}))$ and $v$ is incident to $e_{v_1},...,e_{v_k}$.
        \par The output of the instance, or the value of $G$, is the sum of the values of all possible assignments of $E$, denoted as:
        $$Z(G)=\sum_{\sigma\in \{0,1\}^{|E|}}\omega(\sigma)$$ 
            \label{defHol}
        \end{definition}
        
        Similarly, if each signature in $\mathcal{F}$ is restricted to be symmetric, we denote this kind of problems as symmetric Holant, or sym-Holant for short. If we restrict the maximum degree of the vertices in $V$ to be no more than a constant $k$, we denote this kind of problems as \textit{bounded degree }Holant or \textit{maximum degree }$k$ Holant. 
        
        Furthermore, we use $\text{Holant}(\mathcal{F}_1|\mathcal{F}_2)$ represents $\text{Holant}(\mathcal{F}_1\cup \mathcal{F}_2)$ with the restriction that the underlying graph $G=(U,V,E)$ is bipartite, and each vertex $u\in U$ is assigned a signature from $\mathcal{F}_1$ while each vertex $v\in V$ is assigned a signature from $\mathcal{F}_2$.
         We denote by $\mathcal{EQ}$ the set of all equality functions. In other words, $\mathcal{EQ}=\{=_k|k\ge1\}$ where $=_k$ is the signature $[1,0,...,0,1]_k$. We also denote $\{=_k|1\le k\le D\}$ by $\mathcal{EQ}_{\le D}$ for each integer $D\ge 1$.  
         By definition, we have the following lemma. 
         \begin{lemma}
         Let $\mathcal{C}$ be an arbitrary graph class, $\mathcal{F}$ be an arbitrary signature set and $D\ge 1$ be an integer. Then,
         
             $$\text{\#CSP}(\mathcal{F})\langle \mathcal{C}\rangle\equiv_T\text{Holant}(\mathcal{F}|\mathcal{EQ})\langle \mathcal{C}\rangle$$ 
              $$\#R_D\text{-CSP}(\mathcal{F})\langle \mathcal{C}\rangle\equiv_T\text{Holant}(\mathcal{F}|\mathcal{EQ}_{\le D})\langle \mathcal{C}\rangle$$ 
             \label{lemcsp=hol}
         \end{lemma}

         Besides, standard \#PM is exactly $\text{Holant}(\{[0,1,0,...,0]_k|k\in \mathbb{N}^+\})$. Suppose that $G=(V,E)$ is a graph with each vertex of degree $k$ assigned $[0,1,0,...,0]_k$. For each  assignment $\sigma$ of $E$, $\omega(\sigma)$ can only be either 0 or 1 since the value of each signature can only provide a multiplier with value 0 or 1. Furthermore, $\omega(\sigma)=1$ if and only if all edges assigned the value 1 form a perfect matching of $G$. Consequently, we have \#PM$(G)=Z(G)$. In addition, by exchanging 0 and 1, standard \#PM can also be expressed as $\text{Holant}(\{[0,...,0,1,0]_k|k\in \mathbb{N}^+\})$. Similarly, \#PM is exactly $\text{Holant}(\{[0,...,0,1,0]_k|k\in \mathbb{N}^+\}|\{[1,0,c]|c\in \mathbb{C}\})$ since $\{[1,0,c]|c\in \mathbb{C}\}$ is capable of expressing the weight of each edge.

      Suppose that $E'\subseteq E$, $\tau$ is an assignment of $E'$ and $\sigma$ is an assignment of $E$. If for any $e\in E'$, $\sigma(e)=\tau(e)$, then we say $\sigma$ is an extension of $\tau$, denoted by $\sigma\rhd \tau$. This notation can also be extended to strings, as an assignment can be seen as an input string for the instance. The value of $\tau$ is defined as:
      
        $$\omega(\tau)=\sum_{\sigma\rhd \tau}\omega(\sigma)$$
        We also define an instance of $\text{Holant}(\mathcal{F})$ with underlying graph $G^\tau=(V,E-E')$. For each vertex $v\in G^\tau$, if $f_v$ is assigned to $v$ in $G$, then the signature $f_v^{\tau_v}$ is assigned to $v$ in $G^\tau$, where $\tau_v$ is the assignment of all edges incident to $v$ and belong to $E'$, satisfying $\tau\rhd \tau_v$ . By definition, we have:
        
        $$Z(G^\tau)=\omega(\tau)$$
       The value of $G$ can also be written as follows:
       
        $$Z(G)=\sum_{\sigma\in \{0,1\}^{|E|}}\omega(\sigma)=\sum_{\tau\in \{0,1\}^{|E'|}}\sum_{\sigma\rhd \tau}\omega(\sigma)=\sum_{\tau\in \{0,1\}^{|E'|}}Z(G^{\tau})$$
        
        This form implies the following lemma:
        \begin{lemma}
            Suppose that $G=(V,E)$ is an instance of $\text{Holant}(\mathcal{F})$ and each vertex $v\in V$ is assigned a signature from $\mathcal{F}$. Let $E'\subseteq E$ with $|E'|\le c$, where $c$ is a constant. If for arbitrary $\tau\in \{0,1\}^{|E'|}$, $Z(G^{\tau})$ can be computed in polynomial time, then $Z(G)$ can be computed in polynomial time.
            \label{exhaustsum}
        \end{lemma}

\subsection{Reduction methods}
This section presents three kinds of common methods used in reductions in counting complexity.

\subsubsection{Constructing gadgets}

  A \textit{gadget} of $\text{Holant}(\mathcal{F})$ has an underlying graph $GG=(V,E,D)$, where $E$ is the set of normal edges and $D$ is the set of edges with only one endpoint, called \textit{dangling edges} \footnote{In order to differentiate from the notation of a graph, we use two capital letters to represent a gadget.}. Each vertex in $GG$ is still assigned a signature from $\mathcal{F}$. A signature $f$ of arity $|D|$ is said to be \textit{realized} by $GG=(V,E,D)$, if for each assignment $\alpha : D \to \{0,1\}$, $f(\alpha)=Z(GG^{\alpha})$. In this case, we also say $f$ can be realized by $\mathcal{F}$. Similarly to Lemma \ref{exhaustsum} we have:
        
        \begin{lemma}
            Suppose $c$ is a constant and $GG=(V,E,D)$ is a gadget of $\text{Holant}(\mathcal{F})$ with $|D|\le c$. If for arbitrary  $\tau : D \to \{0,1\}$, $Z(GG^{\tau})$ can be computed in polynomial time, then the signature realized by $GG$ can be computed in polynomial time.
            \label{exhaustgad}
        \end{lemma}
        By constructing gadgets with existing signatures, we are able to realize desired signatures.

        \begin{lemma}
            If $f$ can be realized by $\mathcal{F}$, then $$\text{Holant}(\mathcal{F})\equiv_T\text{Holant}(\mathcal{F}\cup\{f\})$$  
        \end{lemma}
        
        Also, we present some derivative concepts related to the concept of a gadget. A \textit{left-side gadget} of $\text{Holant}(\mathcal{F}_1|\mathcal{F}_2)$ has a bipartite underlying graph $GG=(U,V,E,D)$, where each vertex $u\in U$ is assigned a signature from $\mathcal{F}_1$, each vertex $v\in V$ is assigned a signature from $\mathcal{F}_2$, $E$ is the set of normal edges and $D$ is the set of dangling edges. Furthermore, the endpoint of each dangling edge must belong to $U$. It is easy to verify that, if $f$ can be realized by $GG$, then $\text{Holant}(\mathcal{F}_1|\mathcal{F}_2)\equiv_T\text{Holant}(\mathcal{F}_1\cup\{f\}|\mathcal{F}_2)$. The \textit{right-side gadget} is defined similarly except that the endpoint of each dangling edge must belong to $V$.
        
        Besides, for a graph $G=(V,E)$ and $V'\subseteq V$, a \textit{gadget induced by} $V'$ is a gadget $GG=(V',E_{V'},D_{V'})$, where $E_{V'}=\{uv\in E|u,v\in {V'}\}$ and $D_{V'}$ contains a dangling edge connecting to $u\in {V'}$ for each edge in $\{uv\in E|u\in {V'},v\in V-{V'}\}$.  

        \subsubsection{Polynomial interpolation}

 Another crucial technique for proving $\text{\#P}$-hardness of a problem is polynomial interpolation. This method usually involves the construction of a series of gadgets, which are used to simulate a specific signature. For example, we can prove the following lemma through this method:
        \begin{lemma}
            Suppose that there is a sequence of signatures $f_1,f_2,...,f_k,...$, where $f_k=[1,0,f(k)]$ and $f:\mathbb{N}^+\to \mathbb{C}$ is a computable function satisfying $f(i)\neq f(j)$ whenever $i\neq j$. If for each $k\in \mathbb{N}^+$, $f_k$ can be realized by a gadget $GG_k$ of $\text{Holant}(\mathcal{F})$ of size $poly(k)$, then for any $c\in\mathbb{C}$, 
            
            $$\text{Holant}(\mathcal{F}\cup\{[1,0,c]\})\le_T\text{Holant}(\mathcal{F})$$
        \label{interp}
        \end{lemma}
        \begin{proof}
            For any instance $G$ of $\text{Holant}(\mathcal{F}\cup\{[1,0,c]\})$, suppose that the signature $[1,0,c]$ appears $n$ times in $G$. For an instance $G_x$ obtained by replacing every $[1,0,c]$ signature in $G$ with $[1,0,x]$, where $x$ can be seen as a variable, $p(x)=Z(G_x)$ forms a polynomial of $x$ of degree at most $n$. We now construct an instance $G_k$ of  $\text{Holant}(\mathcal{F})$ for each $1\le k\le n$ by replacing every $[1,0,c]$ signature in $G$ with a gadget $GG_k$. As $GG_k$ is of size $poly(k)$, the size of $G_k$ is of size $poly(n)$ for each $k$. $Z(G_k)$ can be computed for each $1\le k\le n$ by using the oracle that computes every instance in $\text{Holant}(\mathcal{F})$. By definition, $Z(G_k)=p(f(k))$ for each $1\le k\le n$, and since for each $1\le i< j\le n, f(i)\neq f(j)$, the coefficient matrix is Vandermonde and we can solve the coefficients of $p(x)$ using these equations. As $Z(G)=p(c)$, $Z(G)$ can also be computed.
        \end{proof}
       \begin{corollary}\label{corointerp}
           Suppose that $a\in \mathbb{C}$ is not a root of unity. If $[1,0,a]$ can be realized by a gadget $GG_1$ of $\text{Holant}(\mathcal{F})$,  then for any $c\in\mathbb{C}$, 
           
            $$\text{Holant}(\mathcal{F}\cup\{[1,0,c]\})\le_T\text{Holant}(\mathcal{F})$$
       \end{corollary}
       \begin{proof}
           Suppose $LP_k$ is a path of length $k$ and has a dangling edge on each end. If each vertex in $LP_k$ is assigned a $[1,0,a]$ signature, $LP_k$ becomes a gadget of signature $[1,0,a^k]$. Since $a$ is not a root of unity, for each $i,j\in \mathbb{N}^+$ where $i\neq j$, we have $a^i\neq a^j$. Thus by replacing each vertex in $LP_k$ with the gadget $GG_1$, we obtain $GG_k$ satisfying the condition in Lemma \ref{interp}.
       \end{proof}
        \par In the rest part of this article, we sometimes omit the proof of polynomial interpolation if the proof is similar to Lemma \ref{interp}. Instead, we will provide a succinct overview of the proof and highlight the particulars that we should pay attention to.  

\subsubsection{Holographic Transformation}
     Let $T$ be a binary signature, and we denote the two dangling edges corresponding to the input variables of it as a left edge and a right edge. Its value then can be written as a matrix $T=\begin{pmatrix}
          t_{00}&t_{01}\\
          t_{10}&t_{11}
      \end{pmatrix}$, where $t_{ij}$ is the value of $T$ when the value of left edge is $i$ and that of the right edge is $j$.

    This notation is conducive to the efficient calculation of the gadget's value. Let us consider two binary signatures, $T$ and $P$, with the right edge of $T$ connected to the left edge of $P$. $T$ and $P$ now form a binary gadget. Subsequently, it can be demonstrated that the value of the resulting gadget is precisely $TP$, which represents the matrix multiplication of $T$ and $P$.
      
      For a signature $f$ of arity $n$ and a binary signature $T$, we use $Tf/fT$ to denote the signature ``$f$ transformed by $T$'', which is a signature of arity $n$ obtained by connecting the right/left edge of $T$ to every dangling edge of $f$. For a set $\mathcal{F}$ of signatures, we also define $T\mathcal{F}=\{Tf|f\in \mathcal{F}\}$. Similarly we define $\mathcal{F}T$. The following theorem demonstrates the relationship between the initial and transformed problems: 
      \begin{theorem}[{Holographic Transformation}{\cite{valiant2008holographic,cai2007valiant}}]
          $Holant(\mathcal{F}|\mathcal{G}) \equiv_T Holant(\mathcal{F}T^{-1}|T\mathcal{G})$
          \label{thmHT}
      \end{theorem}
      Let $H_2=\begin{pmatrix}
          1&1\\
          1&-1
      \end{pmatrix}$. For a set of signatures $\mathcal{F}$, we use $\widehat{\mathcal{F}}$ to denote $H_2\mathcal{F}$. As $H_2^{-1}=\frac{1}{2}H_2$, by Theorem \ref{thmHT} we have:
      
      $$Holant(\mathcal{F}|\mathcal{G}) \equiv_T Holant(\widehat{\mathcal{F}}|\widehat{\mathcal{G}})$$

     We additionally present the following fact as a lemma for future reference.
     \begin{lemma}
         For each $k\ge 1$, $\widehat{=_k}=[1,0,1,0,1,0,\dots]_k$. For example, $\widehat{=_1}=[1,0], \widehat{=_2}=[1,0,1], \widehat{=_3}=[1,0,1,0]$. Consequently, $\widehat{\mathcal{EQ}}= \{[1,0,1,0,1,0,\dots]_k| k\ge 1\}$ and for an integer $D\ge 1$, $\widehat{\mathcal{EQ}_{\le D}}= \{[1,0,1,0,1,0,\dots]_k| 1\le k\le D\}$.
         \label{transEQ}
     \end{lemma}
\subsection{The complexity of \#CSP}\label{sec:precspdicho}
        We begin with the introduction of some special signature sets.
        \subsubsection{Tractable signature sets}
        \begin{definition}
        A signature $f$ with arity $k$ is of \textit{affine type} if it has the form:
        
        $$\lambda\chi_{AX=b}\cdot \mathfrak{i}^{\sum_{1\le i\le n}a_i{x_i}^2+2\sum_{1\le i< j\le n} b_{ij}x_ix_j}$$
        where $\mathfrak{i}$ denotes the imaginary unit, $a_i,b_{ij}\in\{0,1\},\lambda\in \mathbb{C}$, $AX=b$ is a system of linear equations on $\mathbb{Z}_2$ and
        \begin{equation}
            \chi_{AX=b}=\begin{cases}
                1, & \text{if } AX=b; \\
                0, & \text{otherwise}.\notag
            \end{cases}
        \end{equation}
        \par $\mathscr{A}$ denotes the set of all the signatures of affine type.
        \label{defA}
    \end{definition}
       \begin{definition}
        A signature $f$ is an $\mathcal{E}$\textit{-signature} if it has value 0 except for 2 complementary supports. A signature $f$ with arity $k$ is of \textit{product type} if it can be expressed as a tensor product of $\mathcal{E}$-signatures.
       
        \par $\mathscr{P}$ denotes the set of all the signatures of product type.
         \label{defP}
    \end{definition}
\subsubsection{Permutable matchgate signatures}
Permutable matchgate signatures are highly related to the tractable cases in \#CSP$\langle \mathcal{PL}\rangle$. 
 \begin{definition}
          A matchgate is a planar graph $G=(V,E)$ with weighted edges $w:E\to \mathbb{C}$, and together with some external nodes $U\subseteq V$ on its outer face labeled by $\{1,2,...,|U|\}$ in a clockwise order. The signature $f$ of a matchgate $G$ is a Boolean signature of arity $|U|$ and for each $\alpha\in \{0,1\}^{|U|}$,
          
          $$f(\alpha)=\#PM(G-X),$$
          where $X\subseteq U$ and a vertex in $U$ with label $i$ belongs to $X$ if and only if the $i$th bit of $\alpha$ is 1.
          \par A signature $f$ is a matchgate signature if it is the signature of some matchgate.
           $\mathscr{M}$ denotes all the matchgate signatures.
          \label{defMG}
      \end{definition}

\begin{definition}
    Suppose $f$ is a signature of arity $n$. For a permutation $\pi\in S_n$, we use $f_\pi$ to denote the signature
    
$$f_\pi(x_1,\dots,x_n)=f(\pi(x_1),\dots,\pi(x_n))$$
If for each $\pi\in S_n$, $f_\pi$ is a matchgate signature, we say $f$ is a permutable matchgate signature.

We use $\mathscr{M}_{P}$ to denote the set of all the permutable matchgate signatures.
\label{def:MP}
\end{definition}

The following properties can be easily verified.
  \begin{lemma}\label{lemPin}
        If $f$ is a permutable matchgate signature of arity $k$, then for arbitrary $1\le p\le k$ and $\alpha\in \{0,1\}^p$, $f^\alpha$ is also a permutable matchgate signature.
    \end{lemma}
\begin{lemma}
        A symmetric matchgate signature is also a permutable matchgate signature.
    \end{lemma}
    Matchgate signatures can be simulated in polynomial size.
      \begin{lemma}[\cite{cai2013matchgates}]
          A matchgate signature of arity $k$ can be realized by a matchgate with at most $O(k^4)$ vertices, which can be constructed in $O(k^4)$ time.
          \label{lemMG}
      \end{lemma}
  
    If a signature $f$ takes only non-zero values on supports with even/odd Hamming weight, we say $f$ is of \textit{even/odd parity}. Furthermore, a signature $f$ is said to satisfy the \textit{parity condition} if it is of either even or odd parity. It can be easily verified that the parity condition is closed under tensor production and taking self-loops, and thus closed under gadget construction. In other words, if $f$ is the signature of the gadget $GG$, and each vertex in $GG$ is assigned a signature satisfying the parity condition, then $f$ also satisfies the parity condition. 
    
    Since graphs with an odd number of vertices have no perfect matching, all matchgate signatures satisfy the parity condition. Furthermore, by the construction in \cite[Proposition 6.1 and 6.2]{valiant2008holographic}, we have the following lemma:

\begin{lemma}
    If a signature $f$ of arity $k\le 3$ satisfy the parity condition, then $f$ is a permutable matchgate signature.
    \label{MG<3}
\end{lemma}

    Permutable matchgate signatures are also characterized in detail in \cite{meng2025matchgate}. A \textit{star gadget} is a gadget $ST$ formed by a signature $h$ of arity $n$, and $n$ binary signatures $b_1,\dots,b_n$ where for each $1\le j\le n$, $b_j$ is connected to the $j$th variable of $h$. $h$ is denoted as the \textit{central signature} of $ST$ and $b_j$ is denoted as the \textit{$j$th edge signature}. Figure \ref{fig:stargad} presents the star gadget.
\begin{figure}
            \centering
            \includegraphics[height=0.2\textheight]{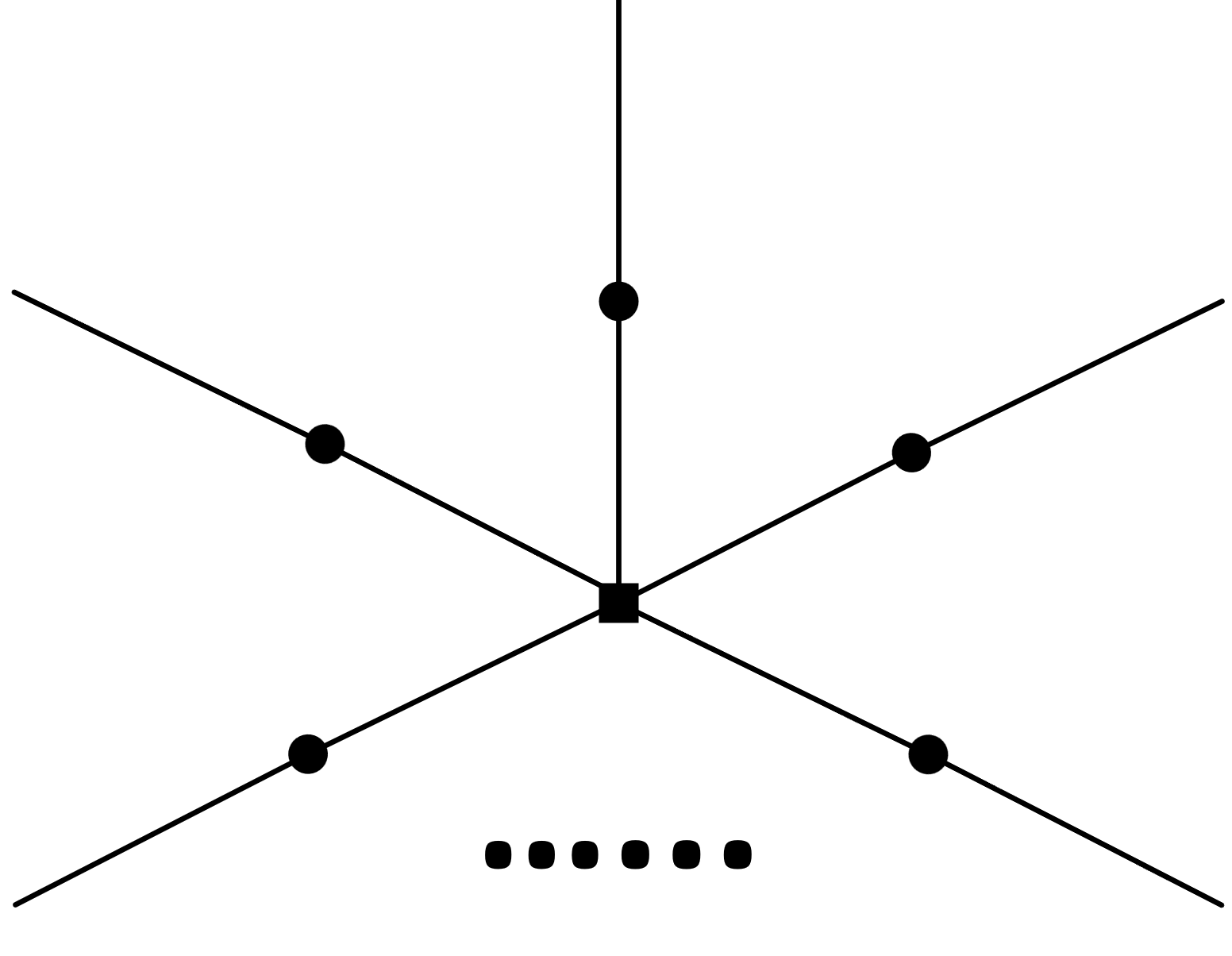}
            \caption{The star gadget. Each edge incident to the vertex of degree $k$, represented by a square, is also incident to a vertex of degree 2 represented by a circle.}
            \label{fig:stargad}
        \end{figure}
 \begin{lemma}\label{lem:SymGadRealize}
    Each permutable matchgate signature $f$ of arity $n$ can be realized by a star gadget, where the center signature is a symmetric matchgate signature and each edge signature is a binary matchgate signature. 
\end{lemma}

\begin{lemma}\label{lem:AsymtoSym}
    For each signature $F\in \mathscr{M}_P-\mathscr{A}$ of arity $n$, a symmetric signature $g\in \mathscr{M}_P-\mathscr{A}$ can be realized by $\{F\}\mid\{[1,0],[1,0,1],[1,0,1,0]\}$ as a planar left-side gadget.
\end{lemma}

\subsubsection{Previous dichotomies for \#CSP}

      The existing dichotomies for \#CSP can be stated as follows:  
\begin{theorem}[\cite{cai2014complexity}]
        If $\mathcal{F}\subseteq \mathscr{A}$ or $\mathcal{F}\subseteq \mathscr{P}$, $\text{\#CSP}(\mathcal{F})$ is computable in polynomial time; otherwise it is $\text{\#P}$-hard.
        \label{genCSP}
    \end{theorem}

\begin{theorem}[\cite{cai2014complexity}]
        Suppose $D\ge 3$ is an integer. If $\mathcal{F}\subseteq \mathscr{A}$ or $\mathcal{F}\subseteq \mathscr{P}$, $\#R_D\text{-CSP}(\mathcal{F})$ is computable in polynomial time; otherwise it is $\text{\#P}$-hard.
        \label{genRDCSP}
    \end{theorem}

    \begin{theorem}[\cite{cai2017holographicuni,meng2025matchgate}]\label{plCSP}
        If $\mathcal{F}\subseteq \mathscr{A}$ or $\mathcal{F}\subseteq \mathscr{P}$ or $\mathcal{F}\subseteq \widehat{\mathscr{M}_P}$, $\text{\#CSP}(\mathcal{F})\langle \mathcal{PL}\rangle$ is computable in polynomial time; otherwise it is $\text{\#P}$-hard.    
    \end{theorem}
\begin{theorem}[\cite{meng2025matchgate}]\label{plRDCSP}
    Suppose $D\ge3$ is an integer. If $\mathcal{F}\subseteq \mathscr{A}$ or $\mathcal{F}\subseteq \mathscr{P}$ or $\mathcal{F}\subseteq \widehat{\mathscr{M}_P}$, $\#R_D$-CSP$(\mathcal{F})\langle \mathcal{PL}\rangle$ is computable in polynomial time; otherwise it is $\text{\#P}$-hard.
\end{theorem}
 At last we present an algorithm result for \#CSP with bounded treewidth from \cite{ganian2018sum}. In the following theorem, \textsf{dom} denotes the size of the domain, \textsf{sup} denotes the maximum size of support in $\mathcal{F}$, $\textsf{tw}^*$ denotes the maximum incidence treewidth and $|I|$ denotes the size of the instance. Furthermore, \#CSPD is a generalized version of \#CSP.

    \begin{theorem}[\cite{ganian2018sum}]
     \#CSPD can be solved in time $(\textsf{dom} +\textsf{sup}+1)^{O(\textsf{tw}^*)}|I|$.
    \label{incidence}
\end{theorem}

In the setting of this article, $\textsf{dom}=2$ and $\textsf{sup}$ is bounded since $\mathcal{F}$ is finite. Consequently,  Theorem \ref{incidence} can induce the following corollary.
   
   \begin{corollary}[\cite{ganian2018sum}]
      For each $k\in \mathbb{N}^+$, if $\mathcal{F}$ is finite, \#CSP$(\mathcal{F})\langle \mathcal{TW}_k \rangle$ is computable in $c^{O(k)}|I|$ time, where $c$ is only parameterized by $\mathcal{F}$.
    \label{CSPK4}
   \end{corollary}
    Furthermore, since $fb(\{K_4\})=\mathcal{TW}_2$, we have that $\text{\#CSP}(\mathcal{F})[K_4]$ can be computed in linear time for arbitrary finite $\mathcal{F}$. This proves the first statement in Theorem \ref{allthm}.

\section{Main results}\label{sec:preprocessing}\label{sec:easyplrdcsp}
In this section, we present  the detailed version of our main results as follows.

\begin{theorem}
For $\text{\#CSP}$ problem over Boolean domain with complex range:
      \begin{enumerate}
          \item For each $k\in \mathbb{N}^+$, \#CSP$(\mathcal{F})\langle \mathcal{TW}_k \rangle$ can be computed in polynomial time for arbitrary finite $\mathcal{F}$; 
          \item Suppose that $H$ is  a single crossing graph. $\text{\#CSP}(\mathcal{F})[H]$ is computable in polynomial time if $\mathcal{F}\subseteq \mathscr{A}$ or $\mathcal{F}\subseteq \mathscr{P}$ or $\mathcal{F}\subseteq \widehat{\mathscr{M}_P}$; otherwise it is $\text{\#P}$-hard. 
          \item  $\text{\#CSP}(\mathcal{F})\langle \mathcal{PLA}\rangle$ is computable in polynomial time if $\mathcal{F}\subseteq \mathscr{A}$ or $\mathcal{F}\subseteq \mathscr{P}$; otherwise it is $\text{\#P}$-hard. 
          \item Suppose that $\mathcal{G}\cap\mathcal{H}\neq \emptyset$, $\mathcal{G}\cap\mathcal{PL}=\emptyset$. Then for any constant $k\in \mathbb{N}^+$,  maximum degree $k$ $\text{\#CSP}(\mathcal{F})[\mathcal{G}]$ is computable in polynomial time if $\mathcal{F}\subseteq \mathscr{A}$ or $\mathcal{F}\subseteq \mathscr{P}$ or $\mathcal{F}\subseteq \widehat{\mathscr{M}_P}$;  otherwise, it is $\text{\#P}$-hard.
          \item Suppose that $\mathcal{G}\cap\mathcal{H}=\emptyset$. Then for any integer $D\ge 3$,  $\#R_D\text{-CSP}(\mathcal{F})[\mathcal{G}]$ is computable in polynomial time if $\mathcal{F}\subseteq \mathscr{A}$ or $\mathcal{F}\subseteq \mathscr{P}$; otherwise, it is $\text{\#P}$-hard.
      \end{enumerate}
      \label{allthm2}
  \end{theorem}
  We remark that Theorem \ref{allthm2} can be applied to \#CSP problems that forbid a complete graph as a minor, thereby obtaining Theorem \ref{allthm}. This application is discussed at the beginning of each corresponding section in Section \ref{secK5}-\ref{secK6}.
In addition, all of the graph classes listed in Theorem \ref{allthm2} are minor-closed classes by Theorem \ref{thmminor}, though it might be challenging to specify the exact forbidden minor sets for some of them.  

In the rest part of this paper, we prove Theorem \ref{allthm2}. The first statement follows from Corollary \ref{CSPK4}.
In order to prove the second to the fifth statements, we address the following analysis, presented as follows.
\begin{remark}\label{remarCSP}
By Theorem \ref{genCSP}, if $\mathcal{F}\subseteq \mathscr{A}$ or $\mathcal{F}\subseteq \mathscr{P}$, then $\text{\#CSP}(\mathcal{F})\langle \mathcal{C} \rangle$ is computable in polynomial time for arbitrary $\mathcal{C}$. By Theorem \ref{plCSP}, if none of $\mathcal{F}\subseteq \mathscr{A}$, $\mathcal{F}\subseteq \mathscr{P}$ and $\mathcal{F}\subseteq \widehat{\mathscr{M}_P}$ holds and $\mathcal{PL}\subseteq\mathcal{C}$, then $\text{\#CSP}(\mathcal{F})\langle \mathcal{C} \rangle$ is $\text{\#P}$-hard. Moreover, it can be verified that $(\widehat{\mathscr{M}}-\mathscr{A})\cap\mathscr{P}=\emptyset$ \cite{cai2017holographicuni}. Consequently, for a graph class $\mathcal{C}$ satisfying $\mathcal{PL}\subseteq\mathcal{C}$, the only unknown case is when $\mathcal{F}\subseteq \widehat{\mathscr{M}_P}$, but $\mathcal{F}\nsubseteq \mathscr{A}$.
         
          By Lemma \ref{lemcsp=hol} and Theorem \ref{thmHT}, $\text{\#CSP}(\mathcal{F})\langle \mathcal{C} \rangle\equiv_T\text{Holant}(\mathcal{F}|\mathcal{EQ})\langle \mathcal{C} \rangle\equiv_T \text{Holant}(\widehat{\mathcal{F}}|\widehat{\mathcal{EQ}})\langle \mathcal{C} \rangle$. As $\widehat{\mathcal{EQ}}=\{[1,0,1,0,...]_k|k\ge1\}$ by Lemma \ref{transEQ}, both $\widehat{\mathcal{F}},\widehat{\mathcal{EQ}}\subseteq \mathscr{M}_P$. Consequently, to show the unknown case is polynomial-time computable, it is sufficient to give an algorithm for Holant problems defined by permutable matchgate signatures.
          
          Besides, as $\widehat{\mathscr{A}}=\mathscr{A}$, there exists a signature $f\in \widehat{\mathcal{F}}$ satisfying $f\in \mathscr{M}_P-\mathscr{A}$. By Lemma \ref{lem:SymGadRealize}, a symmetric signature $g\in \mathscr{M}_P-\mathscr{A}$ can be realized. The possible forms of $g$ is listed in Lemma \ref{M-Aform} as follows.

          For any integer $D\ge 3$, by replacing $\mathcal{EQ}$ with $\mathcal{EQ}_{\le D}$, the same arguments hold for $\#R_D$-CSP  as well.
      \end{remark}
\begin{lemma}[\cite{guo2020complexity}] \label{M-Aform}
    Suppose $g\in \mathscr{M}-\mathscr{A}$ and is symmetric. Then $g$ has one of the following forms.
    \begin{enumerate}
        \item $[0,1,0,...,0]_k,k\ge3$;
        \item $[0,...,0,1,0]_k,k\ge3$;
        \item $[1,0,r], r^4\neq 0,1$;
        \item $[1,0,r,0,r^2,...]_k,k\ge3, r^2\neq 0,1$;
        \item $[0,1,0,r,0,r^2,...]_k,k\ge3, r^2\neq 0,1$.
    \end{enumerate}
\end{lemma}
In the following proofs of the dichotomies, for each graph class $\mathcal{C}$ we consider, we claim that when $\mathcal{F}\subseteq \widehat{\mathscr{M}_P}$ but $\mathcal{F}\nsubseteq \mathscr{A}$, $\text{Holant}(\widehat{\mathcal{F}}|\widehat{\mathcal{EQ}})\langle \mathcal{C} \rangle$ or $\text{Holant}(\widehat{\mathcal{F}}|\widehat{\mathcal{EQ}_{\le D}})\langle \mathcal{C} \rangle$ is either polynomial-time computable or \#P-hard for each  $\mathcal{C}$. 
        
If we claim that $\text{Holant}(\widehat{\mathcal{F}}|\widehat{\mathcal{EQ}})\langle \mathcal{C} \rangle$ or $\text{Holant}(\widehat{\mathcal{F}}|\widehat{\mathcal{EQ}_{\le D}})\langle \mathcal{C} \rangle$ is polynomial-time computable, by  Lemma \ref{lem:AsymtoSym} we know that each signature $g\in\widehat{\mathcal{F}}$ can be realized by a star gadget. Based on this realization, we prove that a polynomial-time algorithm can be obtained by transforming the star gadget into a path gadget and applying a recursive computation based on the specific tree decomposition.

If we claim that $\text{Holant}(\widehat{\mathcal{F}}|\widehat{\mathcal{EQ}})\langle \mathcal{C} \rangle$ or $\text{Holant}(\widehat{\mathcal{F}}|\widehat{\mathcal{EQ}_{\le D}})\langle \mathcal{C} \rangle$ is \#P-hard, by Lemma \ref{lem:AsymtoSym}, for any $f\in \widehat{\mathcal{F}}$ satisfying $f\in \mathscr{M}-\mathscr{A}$, we may use it together with signatures in $\mathcal{EQ}_{\le 3}$ to simulate a symmetric signature $g\in \mathscr{M}-\mathscr{A}$ as a left-side gadget. 
Then, we prove that $\text{Holant}(\{g\}|\widehat{\mathcal{EQ}})\langle \mathcal{C} \rangle$ or $\text{Holant}(\widehat{\mathcal{F}}|\widehat{\mathcal{EQ}_{\le 3}})\langle \mathcal{C} \rangle$ is \#P-hard for each possible form of $g$ presented in Lemma \ref{M-Aform}, by realizing either $[0,0,1,0]$ or $[0,0,0,1,0]$. In the realizations, we always realize left-side gadget in a bipartite planar way,  stated as follows. 

\begin{remark}\label{remGadConstruct}
     In the following, we always construct the left-side gadget of $\text{Holant}(\{g\} |\mathcal{EQ}_{\le 3})$. For future convenience, whenever $uv$ is an edge and both $u$ and $v$ are assigned $g$ in a gadget, we actually mean that we replace $uv$ with $uw,wv$ where $w$ is assigned a $[1,0,1]$ signature in the gadget. Besides, if a gadget is formed by connecting two existing gadgets together, we also automatically replace the connecting edge $uv$ with $uw,wv$, where $w$ is assigned a $[1,0,1]$ signature. These operations would not change the signature of the gadget. Consequently, it can be verified that each obtained gadget always remains a left-side gadget in our following constructions.
\end{remark}

\section{An algorithm for \#CSP on graphs that forbids a single crossing minor}\label{secK5}

In this section, we prove the second statement in Theorem \ref{allthm2}. We restate it as follows.

\begin{theorem}     
    \par For any single crossing graph $H$, $\text{\#CSP}(\mathcal{F})[H]$ is computable in polynomial time if $\mathcal{F}\subseteq \mathscr{A}$ or $\mathcal{F}\subseteq \mathscr{P}$ or $\mathcal{F}\subseteq \widehat{\mathscr{M}_P}$; otherwise it is $\text{\#P}$-hard. 
    \label{CSPK5}
\end{theorem}

The second statement in Theorem \ref{allthm} can be obtained from Theorem \ref{CSPK5} by setting $H=K_5$. Each planar graph forbids a single crossing minor. By Remark \ref{remarCSP}, to prove Theorem \ref{CSPK5}, it is sufficient to prove the following lemma.

\begin{lemma}
    If $\mathcal{F}\subseteq \mathscr{M}_P$, then for any single crossing graph $H$, there is an algorithm that computes $\text{Holant}(\mathcal{F})[H]$ in polynomial time.
     \label{K5alg}
\end{lemma}

\subsection{Overview of the algorithm}\label{K5ccls}

In order to provide a more comprehensive understanding of the algorithm, we will now introduce a number of additional definitions. We pick a node $r\in T$ as the root of the tree. For two nodes $s,t\in T$, we define $s\le t$ to mean that $t$ is an ancestor of $s$. For a node $t\in T$, let $G_{\le t}$ be the graph induced by $\bigcup_{s\le t}\beta(s)$. For an edge $dt\in T$, if $t$ is a child of $d$, we denote $\beta(t)\cap\beta(d)$ as the \textit{navel} of $G_{\le t}$, denoted as $X_t$. The navel of $G_{\le r}$ is defined as $\emptyset$. The algorithm runs as follows:
\begin{enumerate}
    \item Find a normal tree decomposition $(T',\beta')$ with $g-3$-width no greater than $h$, and arbitrarily decide an $r'\in T'$ as the root;
    \item Modify $(T',\beta')$ to another tree decomposition $(T,\beta)$ with a root $r$ and of $g-3$-width no greater than $h$, such that for each $t\in T$, the children of $t$ do not have the same navel.
    \item Compute the representative functions from leaves to the root.
\end{enumerate}

For Step 1, by Theorem \ref{K5trde}, there exist a constant $h$ such that a tree decomposition $(T',\beta')$ of $G$ with $g-3$-width less than $h$ can be computed in $O(|G|^4)$ time for arbitrary $G$. 
We may further assume the tree decomposition is normal as the size of the tree decomposition is at most $O(|G|^4)$ and we can obtain a normal tree decomposition in polynomial time by Lemma \ref{lemnormaltree}.

We noted that we actually compute a tree decomposition $(T',\beta')$ with $--3$-width less than $h$ in Step 1. In contrast, Step 2,3 in the algorithm only need a tree decomposition of bounded $g-3$-width, which means that the algorithm has the potential to solve a wider range of problems. For illustration, let $\mathcal{TWG}_h=\{G|\textsf{p}_{g-3}(G)\le h\}$, and since a normal tree decomposition with bounded $g-3$-width has a polynomial size of expression, our conclusion can be extended to the following form. 

\begin{theorem}     
    For any $h\in \mathbb{N}^+$, $\text{\#CSP}(\mathcal{F})\langle \mathcal{TWG}_h \rangle$ is computable in polynomial time if $\mathcal{F}\subseteq \mathscr{A}$ or $\mathcal{F}\subseteq \mathscr{P}$ ; it is in $P^{NP}$ if $\mathcal{F}\subseteq \widehat{\mathscr{M}_P}$; otherwise it is $\text{\#P}$-hard. 
    \label{CSPTWG}
\end{theorem}

\par We illustrate Step 2 in the following lemma:
\begin{lemma}
    A normal tree decomposition $(T',\beta')$ with $g-3$-width no greater than $h$ and a root $r'$ can be modified to a tree decomposition $(T,\beta)$ with $g-3$-width no greater than $h$ and a root $r$ in polynomial time, such that for each $t\in T$, the children of $t$ do not have the same navel.
    \label{nosamenavelK5}
\end{lemma}
\begin{proof}
    Suppose $(T_1,\beta')$ is a tree decomposition with $g-3$-width no greater than $h$ and a root $r'$. If $a,b\in T_1$ are children of $t$ with the same navel, we delete $bt$ and add $ab$ in $E(T_1)$ to obtain $T_2$.
    $(T_2,\beta')$ is still a tree decomposition with $g-3$-width no greater than $h$ and a root $r'$ by checking Definition \ref{deftw} and \ref{ga3}.
    \par This transformation causes $O(1)$ time. We do this transformation in $T'$ from the root to the leaves successively in a BFS order, and for each node in $T'$, we need at most $O(|T'|)$ transformations to ensure it no longer has 2 children with the same navel, and as a result the total number of transformations needed is at most $O(|T'|^2)$, which can be done in polynomial time. After all the transformation, we obtain $(T,\beta)$ with root $r$ where $V(T)=V(T'),\beta=\beta'$ and $r=r'$, satisfying that for each $t\in T$, the children of $t$ do not have the same navel.
\end{proof}

\par The following subsections will present the methodology for computing the representative functions, as described in Step 3.
\subsection{Path gadgets}
This section deals with the path gadget, which plays an important role in the following algorithm. In the following analysis, for each constant $r$, we use $\sqrt{r}$ to denote the constant $c$ satisfying $c^2=r$ \footnote{The choice of $c$ would not affect our proof.}. We start with the analysis of a symmetric matchgate signature.

\begin{lemma}
    A symmetric matchgate signature must be a multiple of the following forms:
    \begin{enumerate}
        \item $[1,0,0,...,0]$;
        \item $[0,0,...,0,1]$;
        \item $[0,1,0,0,...,0]$;
        \item $[0,0,...,0,1,0]$;
        \item $[1,0,r,0,r^2,0,...],r\neq 0$;
        \item $\sqrt{r}[0,1,0,r,0,r^2,0,...],r\neq 0$.
    \end{enumerate}
    For future convenience, we set $r=1$ when $f$ is of the form 1-4.
    \label{lemMGform}
\end{lemma}

By Definition \ref{defHol}, the constant factor can be ignored, as it only contributes a constant multiplier to the output. Suppose $\alpha\in \{0,1\}^q$ and $f$ is a symmetric matchgate signature of arity $k$. We now consider the form of $f^\alpha$ when $f$ has the form 1-6 in Lemma \ref{lemMGform} \footnote{the notation of $b=0,b=1$ would be explained in the later analyses.}.
\begin{enumerate}
    \item If $HW(\alpha)=0$, $f^\alpha$ has the form $[1,0,0,...,0]$($b=0$); otherwise $f^\alpha$ remains constant at $0$. 
    \item If $HW(\alpha)=q$, $f^\alpha$ has the form $[0,0,...,0,1]$($b=1$); otherwise $f^\alpha$ remains constant at $0$. 
    \item If $HW(\alpha)=0$, $f^\alpha$ has the form $[0,1,0,0,...,0]$($b=0$); else if $HW(\alpha)=1$, $f^\alpha$ has the form $[1,0,0,...,0]$($b=1$); otherwise $f^\alpha$ remains constant at $0$. 
    \item If $HW(\alpha)=q$, $f^\alpha$ has the form $[0,0,...,0,1,0]$($b=1$); else if $HW(\alpha)=q-1$, $f^\alpha$ has the form $[0,0,...,0,1]$($b=0$); otherwise $f^\alpha$ remains constant at $0$. 
    \item If $HW(\alpha)$ is even, $f^\alpha$ has the form $[1,0,r,0,r^2,0,...]$($b=0$); otherwise $f^\alpha$ has the form $\sqrt{r}[0,1,0,r,0,r^2,0,...]$($b=1$). Here, the value of $r$ is consistent with that of $f$.
    \item If $HW(\alpha)$ is odd, $f^\alpha$ has the form $[1,0,r,0,r^2,0,...]$($b=1$); otherwise $f^\alpha$ has the form $\sqrt{r}[0,1,0,r,0,r^2,0,...]$($b=0$).  Here, the value of $r$ is consistent with that of $f$.
\end{enumerate}

If $f^\alpha$ does not remain constant at $0$, we call $\alpha$ a \textit{valid} pinning of $f$. Notice that, for each form of $f$, there are at most 2 forms that $f^\alpha$ can take under a valid pinning. The form that $f^\alpha$ takes consequently can be expressed as a single bit $b$, denoted as an \textit{information bit} $b$ of $\alpha$ over $f$. The value of the information bit in each case is listed in the parentheses following the previous analyses. As a result, for a signature $f$ of arity $k$, there exists a function $Q_{(f,q)}:\{0,1\}^q\to\{0,1\}$ for each integer $1\le q\le k-1$ such that it maps a valid pinning $\alpha$ of a symmetric matchgate signature $f$ to the information bit $b$ of $\alpha$ over $f$. 
 
 We can verify that for each $\alpha=\alpha_1\alpha_2$, if $f(\alpha)\neq 0$, then $Q_{(f,q_1)}(\alpha_1)= Q_{(f,q_2)}(\alpha_2)$ if $f$ is of the form 1,2 or 5, and we denote these 3 forms as type 1; $Q_{(f,q_1)}(\alpha_1)\neq Q_{(f,q_2)}(\alpha_2)$ if $f$ is of the form 3,4 or 6, and we denote these 3 forms as type 2.
\par Since for each $\alpha\in\{0,1\}^q$, the form of $f^\alpha$ can be decided by a single bit $Q_{f,q}(\alpha)$, it is evident that replacing the $q$ bits with a single bit in $f$ would result in a decrease in the arity of $f$ while maintaining the form of $f^\alpha$ for each $\alpha$. The following definition is based on this line of reasoning:
\begin{definition}
    A \textit{$q$-information signature} of a symmetric matchgate signature $f$ of arity $k$, denoted as $f_q$, is a matchgate signature of arity $k-q+1$ depending on variables $(x_b,x_{q+1},...,x_k)$, where $x_b$ is denoted as an \textit{information variable} corresponding to the information bit. For each $f$ of the different forms in Lemma \ref{lemMGform},  $f_q$ has the following form:
\begin{enumerate}
    \item $f_q=[1,0,0,...,0]$;
    \item $f_q=[0,0,...,0,1]$;
    \item $f_q=[0,1,0,0,...,0]$;
    \item $f_q=[0,0,...,0,1,0]$;
    \item $f_q^{x_b=0}=[1,0,r,0,r^2,0,...]$, $f_q^{x_b=1}=\frac{1}{\sqrt{r}}[0,r,0,r^2,0,...]$;
    \item $f_q^{x_b=0}=\sqrt{r}[0,1,0,r,0,r^2,0,...]$, $f_q^{x_b=1}=[1,0,r,0,r^2,0,...]$.
\end{enumerate}
    \label{definfsig}
\end{definition}

When $f$ is of the form 5(or 6) in Lemma \ref{lemMGform}, $f_q$ is not symmetric; however, it can be realized by symmetric matchgate signatures by connecting a dangling edge of a single $[1,0,\frac{1}{\sqrt{r}}]$ to a dangling edge of the $[1,0,r,0,r^2,0,...]_{k-q+1}$(or $\sqrt{r}[0,1,0,r,0,r^2,0,...]_{k-q+1}$) signature. 
\par $f_q$ can represent the behaviors of $f^\alpha$ for all valid pinnings $\alpha\in \{0,1\}^q$. It can be verified that, if $f(\alpha)\neq 0$, then $f(\alpha)=\sqrt{r}^{HW(\alpha)}$. Consequently, for each $\gamma\in \{0,1\}^{k-q}$, 
$$f_q(Q_{(f,q)}(\alpha),\gamma)=(\frac{1}{\sqrt{r}})^{HW(\alpha)}f^\alpha(\gamma)$$

The following lemma can also be verified:
\begin{lemma}
Suppose $f$ is a symmetric matchgate signature of arity $k$ and $\alpha\in \{0,1\}^k$ is an assignment string satisfying $f(\alpha)\neq 0$.
Furthermore, $\alpha_a\in \{0,1\}^q,\alpha_b\in \{0,1\}^{k-q}$ and their linkage $\alpha_a\alpha_b=\alpha$. Then the following equation holds:
$$f(\alpha)=f_{q}(Q_{(f,q)}(\alpha_a),\alpha_b)\cdot f_{k-q}(Q_{(f,k-q)}(\alpha_a),\alpha_b)$$

    \label{MGrepre}
\end{lemma}
\begin{proof}
    \begin{align}
        & f_{q}(Q_{(f,q)}(\alpha_a),\alpha_b)\cdot f_{k-q}(Q_{(f,k-q)}(\alpha_a),\alpha_b) \notag \\
        ={}& (\frac{1}{\sqrt{r}})^{HW(\alpha_a)}f^{\alpha_a}(\alpha_b)\cdot(\frac{1}{\sqrt{r}})^{HW(\alpha_b)}f^{\alpha_b}(\alpha_a)\notag \\
        ={}& (\frac{1}{\sqrt{r}})^{HW(\alpha)}(f(\alpha))^2\notag \\
         ={}& f(\alpha)\notag     
    \end{align}
\end{proof}

   %
For each symmetric matchgate signature $f$, we can construct a path gadget as follows.
\begin{definition}
    Suppose $f$ is a symmetric matchgate signature of arity $k$. For $1\le q\le k-1$, a $q$-path gadget $PP_q=(V,E,D)$ of $f$ is defined as follows. $PP_q$ is a path $uxwyv$, with $q$ dangling edges connecting to $u$ and $k-q$ edges connecting to $v$. $w$ is assigned $[1,0,1]$ if $f$ is of type 1; otherwise it is assigned $[0,1,0]$. $ux$ is the gadget that realize $f_{k-q}$ while $yv$ is the gadget that realize $f_{q}$. In other words, $x,y$ are assigned $[1,0,\frac{1}{\sqrt{r}}]$ and  $u$ is assigned $f_{k-q}$ if $f$ is of form 1-4, $[1,0,r,0,r^2,0,...]_{q+1}$ (or $\sqrt{r}[0,1,0,r,0,r^2,0,...]_{q+1}$) if $f$ is of the form 5 (or 6). The signature assigned to $v$ is similar to that to $u$, except that it is of arity $k-q+1$.

    We also denote $u$ as the head vertex of $PP_q$ and $v$ as the tail vertex of $PP_q$.
\end{definition}
 See figure \ref{fig:K54} for a visualization of the path gadget. Suppose $f\in\mathscr{M}_P$. By Lemma \ref{lem:SymGadRealize}, $f$ can be realized by a star gadget, and the central signature can further be realized by a path gadget. Each end of the path, together with the edge signatures connecting to it in the star gadget, again form a gadget whose signature is a permutable matchgate signature. Using this observation,  we further extend the concept of the path gadget to $f\in\mathscr{M}_P$.
 
\begin{definition}
    Suppose $f\in\mathscr{M}_P$ is of arity $k$ and $(S_1,S_2)$ is a partition of $\{1,2,\dots,k\}$. Let $ST_f$ be the star gadget that realize $f$ as described in Lemma \ref{lem:SymGadRealize} and $h_f$ be the symmetric central signature. For $q=|S_1|$, let $PP_q=(V,E,D)=uxwyv$ be the $q$-path gadget of $h_f$, and we replace $h_f$ with the $PP_q$ gadget in $ST_f$ to obtain the gadget $ST_f'$, such that variables with indices in $S_1$ are connecting to the head of $PP_q$, while those corresponding to $S_2$ are connecting to the tail.
    
    A $(S_1,S_2)$-path gadget $PP_{(S_1,S_2)}=(V,E,D)$ of $f$ is defined as follows. The underlying graph of $PP_{(S_1,S_2)}$ is same as $PP_q$. $x,w,y$ are assigned the same signature as in $ST_f'$. $u$ is assigned the signature of the gadget formed by the $q$ edge signatures with indices in $S_1$ and the signature assigned to the head vertex of $PP_q$ in $ST_f'$. In other words, it is the signature of the gadget induced by $\{u\}\cup N_u-x$ in the gadget $ST_f'$, where $N_u$ denotes the neighbours of $u$ in $ST_F'$ and each vertex in $N_u-x$ is assigned the edge signature with index belonging to $S_1$. The signature assigned to $v$ is analogous to that assigned to $u$, except that it is of arity $k-q+1$ and the indices of its variables correspond to $S_2$.

    We denote $u$ as the head vertex of $PP_{(S_1,S_2)}$ and $v$ as the tail vertex of $PP_{(S_1,S_2)}$ in this case as well.
    \label{defpath}
\end{definition}
 
 The following lemma can be easily verified by Lemma \ref{MGrepre}.
 \begin{lemma}
     Suppose $f\in\mathscr{M}_P$ is of arity $k$ and $(S_1,S_2)$ is a partition of $\{1,2,\dots,k\}$. Then a $(S_1,S_2)$-path gadget of $f$ realizes $f$. Furthermore, the signatures assigned to the head vertex and the tail vertex of the gadget are both permutable matchgate signatures.
     \label{lempath}
 \end{lemma}
 

\subsection{Computing the leaves}

In this section, for $dt\in T$ where $t$ is a leaf and a child of $d$, we demonstrate the method for replacing $G_{\le t}$ with a single vertex in $G_{\le d}$. 
\begin{definition}
    \par For $G_{\le t}=(V_{\le t},E_{\le t})$ where $t\in T$ is a leaf, the \textit{representative signature} $f_{\le t}$ of it is defined as the signature of $GG_{\le t}$, and $GG_{\le t}$ is constructed as follows: 

    $GG_{\le t}=(V_{\le t},E_{\le t}-E_{X_t},D_{\le t})$, where $E_{X_t}=\{xy\in E_{\le t}|x,y\in X_t\}$ and $D_{\le t}=\{x|x\in X_t\}$.
    The matchgate signatures assigned to vertices in  $GG_{\le t}-X_t$  remain the same as in $G_{\le t}$. For each $x\in X_t$, suppose a permutable matchgate signature $f$ of arity $k$ is assigned to $x$ in $G_{\le t}$, $S_1=\{s|\text{the }s\text{th edge of }x\in E_{\le t}-E_{X_t}\}$, and $(S_1,S_2)$ is a partition of $\{1,\dots,k\}$. Then in $GG_{\le t}$, $x$ is assigned the signature same as that assigned to the head vertex of the $(S_1,S_2)$-path gadget of $f$.
    \label{defrepsig}
\end{definition}
\begin{figure*}
	\centering
    \subcaptionbox{\label{fig:K51}}{ \includegraphics[width=0.4\textwidth]{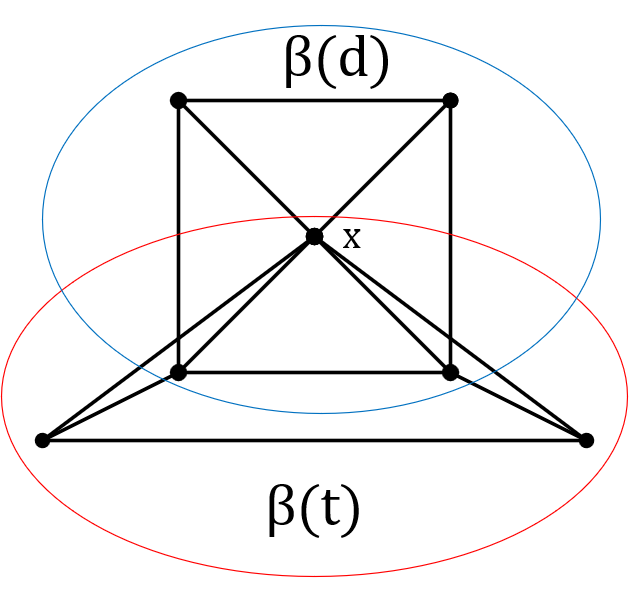}}
		\subcaptionbox{\label{fig:K52}}{\includegraphics[width=0.4\textwidth]{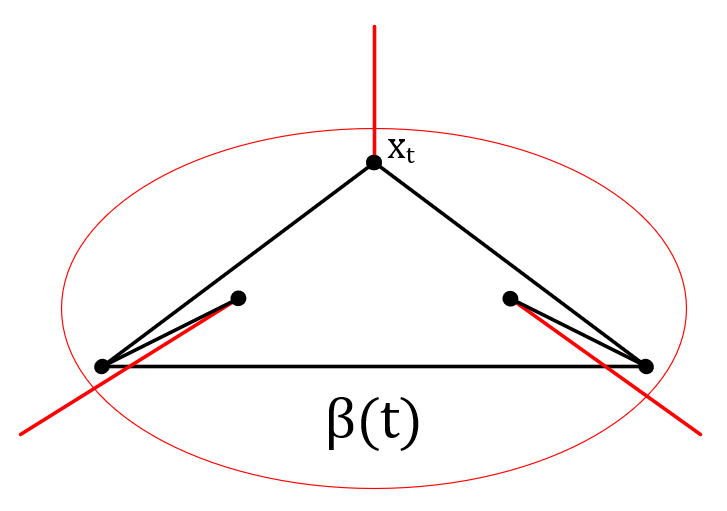}}
  
 \subcaptionbox{\label{fig:K54}}{\includegraphics[width=0.4\textwidth]{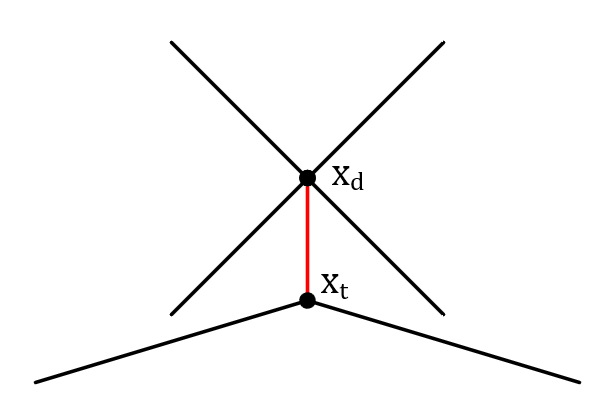}}
 \subcaptionbox{\label{fig:K53}}{\includegraphics[width=0.4\textwidth]{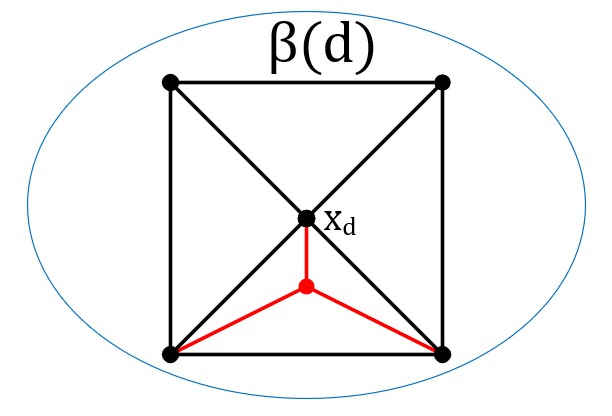}}
	\caption{An example of the procedure of the algorithm. (a) The graph induced by $\beta(t)\cup\beta(d)$ where $t$ is a leaf and a child of $d$. $\beta(t)$ is enclosed by the red ellipse while $\beta(d)$ is enclosed by the blue ellipse. (b) The construction of the gadget $GG_{\le t}$, where the red edges represent the dangling edges in $D_{\le t}$. (c)The path gadget for $x$, where the path is colored in red and 3 vertices of degree 2 on the path are omitted. (d) Visualization of how $G_{\le t}$ is replaced by a single vertex in $G_{\le d}$. The introduced vertex is colored in red as well as the remaining parts of the path gadgets.} 
	\label{fig:K5alg}
\end{figure*}
\par See Figure \ref{fig:K51},\ref{fig:K52} for an example. In the following, we illustrate how to compute the representative signature of $G_{\le t}$. 
\begin{lemma}
    Suppose $G$ is a graph which can be embedded on a surface with genus at most $h$, and each vertex of $G$ is assigned a permutable matchgate signature. Then the value of $G$ can be computed in polynomial time.
    \label{mggealg}
\end{lemma}
\begin{proof}
    After replacing each vertex in $G$ with the corresponding matchgate, the obtained graph $G'$ can also be embedded on a surface with genus at most $h$. This is because for each permutable matchgate signature, all the edges incident to it always form a proper clockwise order for the embedding. This property ensures that no crossings would be introduced in the replacement. By definition we have $Z(G)=$\#PM$(G')$, which is polynomial time computable by \cite{galluccio1999theory,tesler2000matchings,curticapean2015parameterizing} as mentioned in Section \ref{PMintro}. 
\end{proof}

\begin{lemma}
     Suppose $t\in T$ is a leaf, then the representative signature $f_{\le t}$ of $G_{\le t}$ can be computed in polynomial time.
     \label{computerepre}
\end{lemma}
\begin{proof}
    Since the $g-3$-width of $(T,\beta)$ is less than $h$, either $|\beta(t)|\le h$, or the gadget can be embedded on a surface with genus at most $h$.
   For each assignment $\sigma$ of $D_{\le t}$, if $|\beta(t)|\le h$, $|E_{\le t}|\le h^2$. Since $h$ is a constant, $Z(GG_{\le t}^\sigma)$ can be computed in $O(1)$ time by Lemma \ref{exhaustsum}.
    Otherwise, each vertex in $GG_{\le t}^\sigma$ is assigned a permutable matchgate signature by Lemma \ref{lemPin}. 
   Then by Lemma \ref{mggealg}, $Z(GG_{\le t}^\sigma)$ can also be computed in polynomial time. 
   In either case, $Z(GG_{\le t}^\sigma)$ can be computed in polynomial time. As $|D_{\le t}|\le 3$, $f_{\le t}$ can be computed in polynomial time by Lemma \ref{exhaustgad}.
\end{proof}
\subsection{Transforming into matchgates}

As the representative signature is a gadget of signatures satisfying the parity condition, itself also satisfies the parity condition and thus is a matchgate signature by Lemma \ref{MG<3}.

Suppose that $dt\in T$ and $t$ is a leaf of $T$, $X_t$ is the navel of $t$ and $f_{\le t}$ is the representative signature of $G_{\le t}$. In this section, we perform two transformations on $G_{\le d}$, and shows that both transformations would not change the value of the representative signature $f_{\le d}$ of $G_{\le d}$. 

First, for each $x\in X_t$ assigned a permutable matchgate signature $f$ of arity $k$, let $S_1=\{s|\text{the }s\text{th edge of }x\in E_{\le t}-E_{X_t}\}$, and $(S_1,S_2)$ be a partition of $\{1,\dots,k\}$. We replace $x$ with a $(S_1,S_2)$-path gadget $PP_{(S_1,S_2)}$. By Lemma \ref{lempath} this transformation doesn't change the value of $f_{\le d}$. We denote the head vertex of $PP_{(S_1,S_2)}$ as $x_t$ and the tail vertex as $x_d$.

Then, vertices in $(\beta(t)-\beta(d))\cup\{x_t|x\in X_t\}$ induce the gadget $GG_{\le t}$. We replace them with a single vertex assigning $f_{\le t}$, which would not change the value of $f_{\le d}$ as well. Besides, $f_{\le t}$ can be computed in polynomial time by Lemma \ref{computerepre}. See Figure \ref{fig:K5alg} for an example of the transformations. We record this result informally as a lemma for future reference.

\begin{lemma}
    Suppose that $dt\in T$ and $t$ is a leaf of $T$. The effect of $G_{\le t}$ can be represented by introducing a single vertex in $G_{\le d}$.
    \label{BagtoMG}
\end{lemma}

Now we show that this procedure can be done successively. For a $d\in T$, if $|\beta(d)|\le h$, there is at most $(h+2)^3$ possible navels since the size of the navel is bounded by 3. By Lemma  \ref{nosamenavelK5} and \ref{lempath}, at most $10(h+2)^3$ vertices can be introduced in $G_d$ in the process in Lemma \ref{BagtoMG}, and consequently the size of $G_d$ is bounded by $10(h+2)^3+h$, which is still a constant. 

Otherwise, $G_d$ can be embedded on a surface with genus at most $h$. After the process in Lemma \ref{BagtoMG}, the obtained graph is still able to be embedded on a surface with genus at most $h$. This is because for each introduced matchgate of arity 1 or 2, no crossing is introduced. For each introduced matchgate of arity 3, it is bounded by a face and each face contains at most one such signature by Lemma \ref{nosamenavelK5}. As a result, the obtained graph after the process in Lemma \ref{BagtoMG} is still able to be embedded on a surface with genus at most $h$. Furthermore, a matchgate signature of arity no greater than 3 is always permutable.

Consequently, after the process in Lemma \ref{BagtoMG}, either $|\beta(d)|\le O(h^3)$, or the gadget $GG_{\le d}$ can be embedded on a surface with genus at most $h$. Besides, all the signatures we introduce in the process are permutable matchgate signatures. Lemma \ref{computerepre} can still be adapted by setting $10(h+2)^3+h$ as the constant we need in the proof, and as a result this procedure can be performed iteratively until the desired output is achieved.

\section{An algorithm for bounded degree \#CSP on minor-free graphs}\label{secK7}

In this section, we prove the fourth statement in Theorem \ref{allthm2}, which can be restated as follows.

\begin{theorem}     
   If $\mathcal{G}\cap\mathcal{H}\neq \emptyset$, $\mathcal{G}\cap\mathcal{PL}= \emptyset$, then for any constant $k\in \mathbb{N}^+$, maximum degree $k$ $\text{\#CSP}(\mathcal{F})[ \mathcal{G} ]$ is computable in polynomial time if $\mathcal{F}\subseteq \mathscr{A}$ or $\mathcal{F}\subseteq \mathscr{P}$ or $\mathcal{F}\subseteq \widehat{\mathscr{M}_P}$; otherwise it is $\text{\#P}$-hard. 
    \label{CSPK7}
\end{theorem}

By Lemma \ref{RD=bounded}, the following corollary can be obtained from Theorem \ref{CSPK7}.

\begin{corollary}
    If $\mathcal{G}\cap\mathcal{H}\neq \emptyset$, $\mathcal{G}\cap\mathcal{PL}= \emptyset$, then for any constant $D\in \mathbb{N}^+$, $\#R_D\text{CSP}(\mathcal{F})[ \mathcal{G} ]$ is computable in polynomial time if $\mathcal{F}\subseteq \mathscr{A}$ or $\mathcal{F}\subseteq \mathscr{P}$ or $\mathcal{F}\subseteq \widehat{\mathscr{M}_P}$; otherwise it is $\text{\#P}$-hard. 
\end{corollary}

The fourth statement in Theorem \ref{allthm} can be obtained from Theorem \ref{CSPK7} as $K_7$ is a minor of $H_{18}$ and can not be embedded on the plane \cite{thilikos2022killing}. As $\mathcal{G}\cap\mathcal{PL}= \emptyset$, all planar graphs are $\mathcal{G}$-minor-free. By Remark \ref{remarCSP}, to prove Theorem \ref{CSPK7}, it is sufficient to prove the following lemma.

\begin{lemma}
     \par Suppose that $\mathcal{F}\subseteq \mathscr{M}_P$ and all the signatures in $\mathcal{F}$ are of arity at most $k$, where $k\in \mathbb{N}^+$ is a constant. Then for any graph class $\mathcal{G}$ satisfying $\mathcal{G}\cap\mathcal{H}\neq \emptyset$, there is an algorithm that computes $\text{Holant}(\mathcal{F})[ \mathcal{G} ]$ in polynomial time.
     \label{K7alg}
\end{lemma}
\subsection{Overview of the algorithm}
\par The algorithm proceeds as follows:
\begin{enumerate}
    \item Find a tree decomposition $(T,\beta)$ with $ga3$-width $\le h$, decide an $r\in T$ as the root;
    \item Compute boundary mappings from leaves to the root.
\end{enumerate}
\par 

\par The algorithm is analogous to that described in Section 3. The algorithm was developed on the basis of the following principle: given that the degree of each vertex in $G$ is bounded, it is now possible to deal with graphs with finite apex vertices, which enables the development of a more generalised algorithm.

We now focus on Step 1. For a graph $H\in\mathcal{G}\cap\mathcal{H}$, let $h(G)$ be the minimum integer $t$ such that $H$ is a minor of $H_t$. Let $h'=\min_{H\in\mathcal{G}\cap\mathcal{H}} h(G)$, which is computable. For each graph $G\in fb(\mathcal{G})$, $G$ does not contain a $H_{h'}$ minor, otherwise a contradiction would ensue. Consequently, $\text{svg}(G) \le h'$. By Theorem \ref{thmmtotw}, $\textsf{p}_{ga3}(G)\le h$ where $h=f(h')$ and  $f$ is a computable function. A tree decomposition $(T,\beta)$ of $G$ with $ga3$-width less than $h$ can also be computed in polynomial time by Theorem \ref{K7trde}. We may further assume the tree decomposition is normal as the size of the tree decomposition we obtain is of polynomial size, and thus we can use polynomial time to obtain a normal tree decomposition by Lemma \ref{lemnormaltree}.
 The following subsections will present the algorithm that computes the boundary mappings, as described in Step 2.
\subsection{Starting from the leaves}


\par Again, we introduce some more definitions to illustrate Step 2. 
For an edge $dt\in T$, if $t$ is a child of $d$, we call $\beta(t)\cap\beta(d)-A_d$ the \textit{navel} of $G_{\le t}$, denoted as $X_t$ \footnote{The definition of navel in Section \ref{secK5} can be seen as a special case in which $A_d=\emptyset$.}. The \textit{boundary edge set} $B_t$ of $G_{\le t}$ is all the edges in $G_{\le t}$ with one endpoint in $\beta(t)\cap \beta(d)$ and the other one in $V(G_{\le t})-\beta(t)\cap \beta(d)$. In particular, both $X_r$ and $B_r$ are defined as $\emptyset$.  We also use $Q_t$ to denote the set of all the edges incident to some $a\in A_t$ and use $P_t$ to denote $Q_t-B_t$. As the maximum degree is bounded by $k$ and the size of $X_t\cup A_t$ is bounded by $h+3$, $|B_t\cup P_t|\le k(h+3)$. 

For $dt\in T$ and $t$ is a child of $d$, we want to record the value of $G_{\le t}$ when the assignment of edges in $B_t$ is given. This would allow us to care only about the value of edges in $B_t$ instead of the inner structure of $G_{\le t}$. 

\begin{definition}
    \par For $dt\in T$ and $t$ is a leaf, the boundary gadget of $G_{\le t}$ is denoted as $GG_{\le t}$ which is the gadget induced by $V(G_{\le t})-\beta(t)\cap \beta(d)$. The \textit{boundary mapping} of $G_{\le t}$ is a mapping $BM_t:\{0,1\}^{|B_t|}\to \mathbb{C}$ such that for each $\sigma\in \{0,1\}^{|B_t|}$, $BM_t(\sigma)=Z(GG_{\le t}^\sigma)$.
\end{definition}
    
The computation of the boundary mapping is based on the following lemma:
\begin{lemma}
    Suppose $G$ is a graph with maximum degree $k$, and there exists $A\subseteq V(G)$ with $|A|\le h$, such that $G-A$ can be embedded on a surface with genus at most $h$, and each vertex of $G$ is assigned a permutable matchgate signature, then the value of $G$ can be computed in polynomial time. 
    \label{mggeapalg}
\end{lemma}

\begin{proof}
    By exhausting all the subsets of $V(G)$ with size no more than $h$, we can find the expected $A$ in polynomial time. We denote all edges incident to a vertex in $A$ as an edge set $D$, and we have $|D|\le hk$.
    
    For each assignment $\sigma\in \{0,1\}^{|D|}$ of the edges in $B$, $G^\sigma$ is the graph $G-A$ with $|A|$ isolated vertices, and consequently can be embedded on a surface with genus at most $h$. By Lemma \ref{lemPin}, the signature assigned to each vertex in $G^\sigma$ is still assigned a permutable matchgate signature. By Lemma \ref{mggealg}, the value of $G^\sigma$ can be computed in polynomial time. Since $|D|\le hk$, the value of $G$ can also be computed in polynomial time by Lemma \ref{exhaustsum}. 
\end{proof}
\par We now deal with computing the boundary mapping of a leaf:
\begin{lemma}
     Suppose $t\in T$ is a leaf, then the boundary mapping of $G_{\le t}$ can be computed in polynomial time.
     \label{computebdm}
\end{lemma}
\begin{proof}
    For each $\sigma\in \{0,1\}^{|B_t|}$ which is an assignment of edges in $B_t$, the signature assigned to each vertex in $GG_{\le t}^\sigma$ is a permutable matchgate signature by Lemma \ref{lemPin}.
   By Lemma \ref{mggeapalg},  $BM_t(\sigma)$ can be computed in polynomial time. Thus, the boundary mapping of $G_{\le t}$ can be computed in polynomial time by Lemma \ref{exhaustgad} since $|B_t|\le k(h+3)$.
\end{proof}

\subsection{Transforming into matchgates}
The objective of this section is to demonstrate the proof of the following lemma, which serves to illustrate that the boundary mappings can be computed in a successive manner. 
\begin{lemma}
    Suppose $t\in T$ and $l_1,...,l_p$ are all of its children. Given the boundary mappings of $G_{\le l_1},...,G_{\le l_p}$, the boundary mapping of $G_{\le t}$ can be computed in polynomial time. 
    \label{succK7}
\end{lemma}
    
    Lemma \ref{succK7} can be seen as an extension of Lemma \ref{computebdm}. To prove this lemma, we first deal with children of $t$ with the same navel $X$.

\begin{figure*}
	\centering
    \subcaptionbox{\label{fig:K71}}{\includegraphics[width=0.35\textwidth]{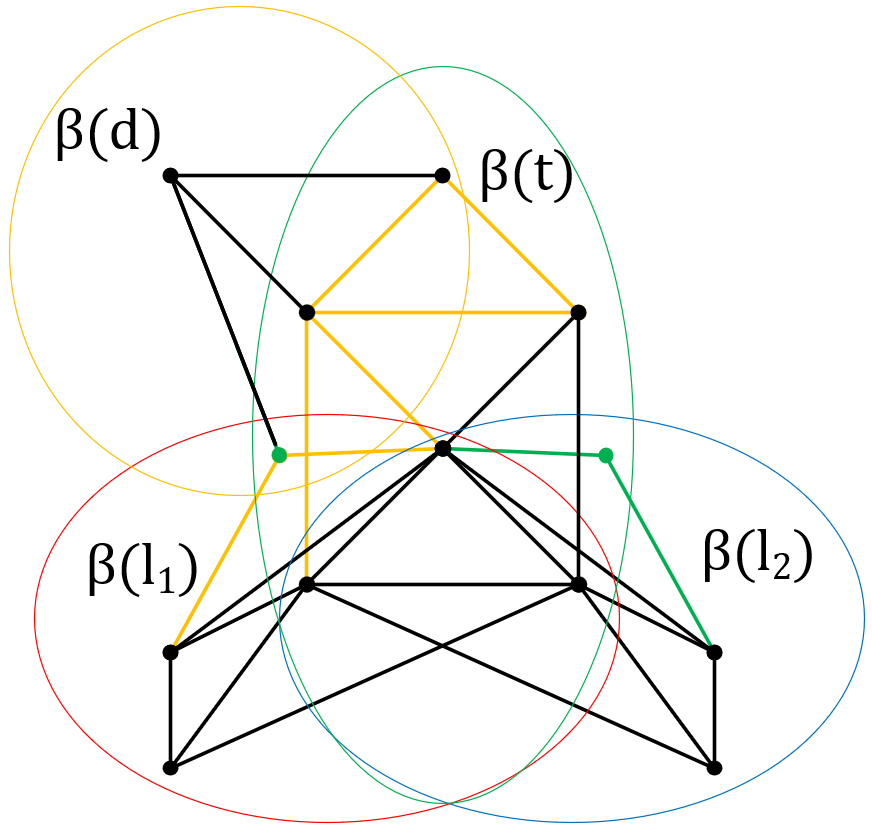}}
		\subcaptionbox{\label{fig:K72}}{\includegraphics[width=0.35\textwidth]{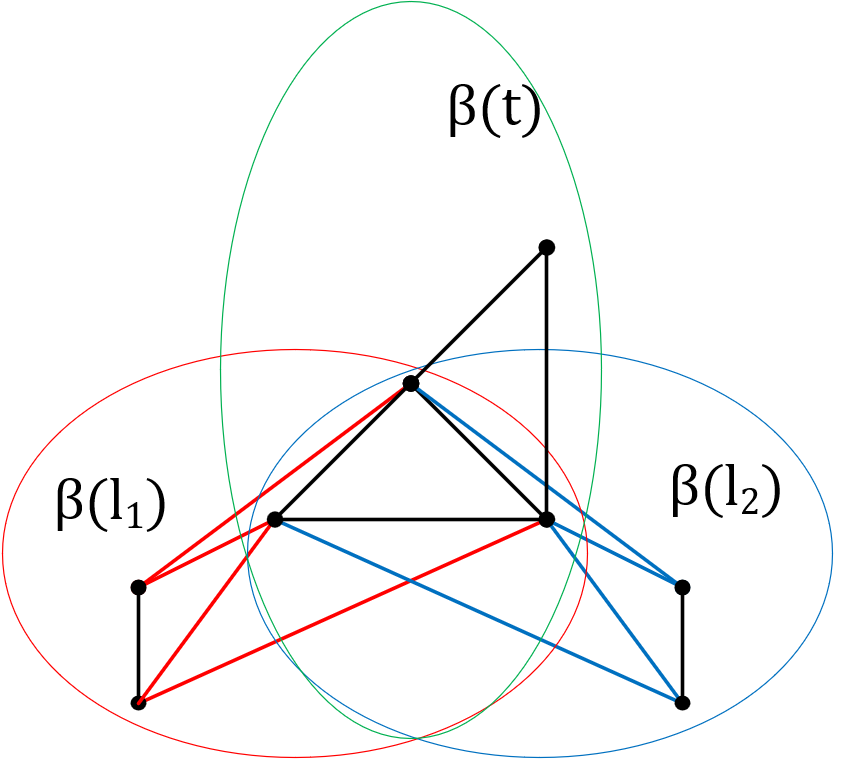}}
  \subcaptionbox{\label{fig:K74}}{\includegraphics[width=0.2\textwidth]{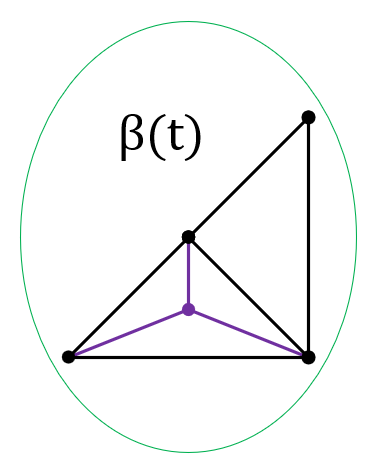}}
  \quad
 
 \subcaptionbox{\label{fig:K73}}{\includegraphics[width=0.5\textwidth]{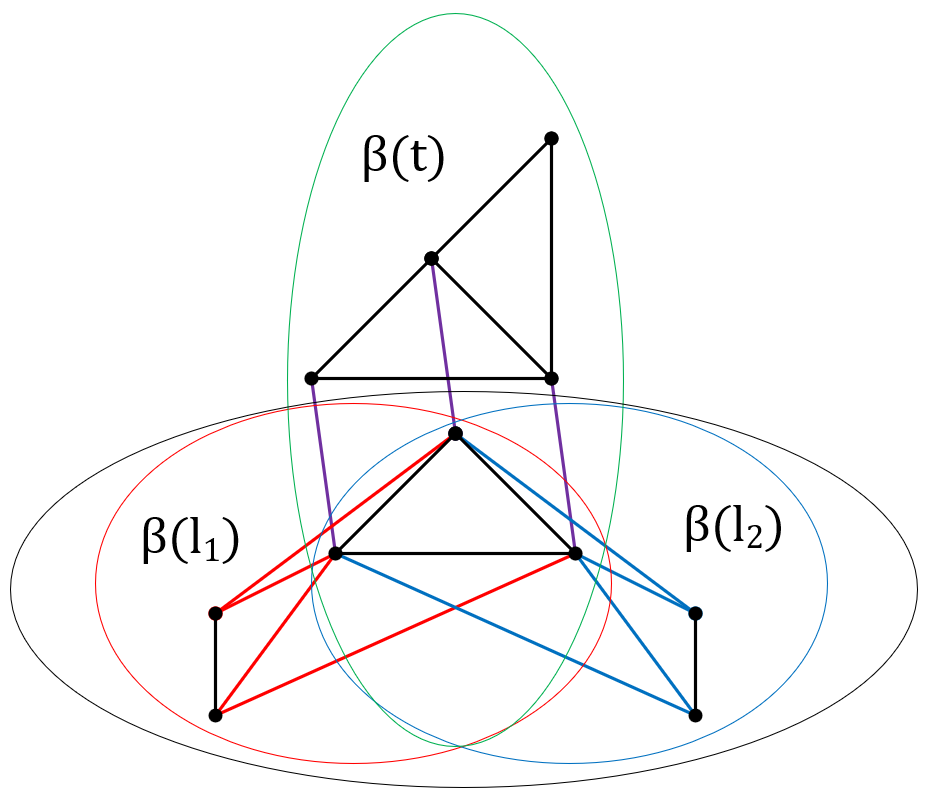}}
 \subcaptionbox{\label{fig:K75}}{\includegraphics[width=0.4\textwidth]{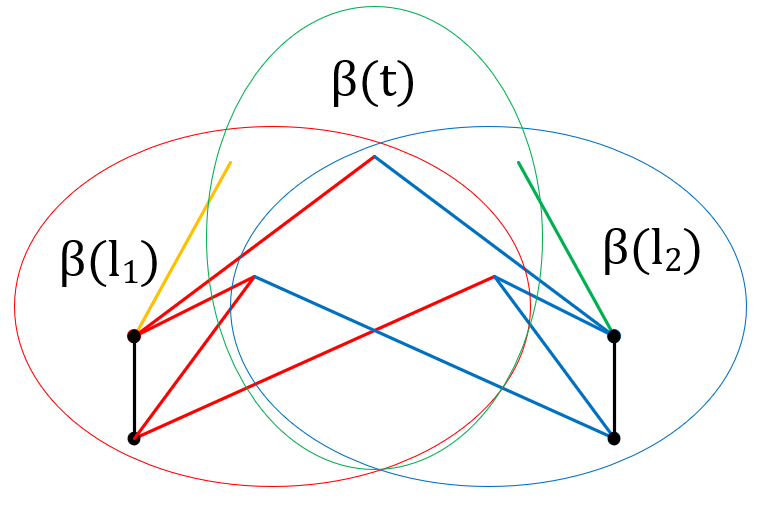}}
 \caption{An example of the procedure of the algorithm. (a) The graph induced by $\beta(d)\cup\beta(t)\cup\beta(l_1)\cup\beta(l_2)$. $\beta(d),\beta(t),\beta(l_1),\beta(l_2)$ are enclosed by the respective orange, green, red, and blue ellipses. Edges in $B_t$ are coloured in orange, while vertices in $A_t$ and edges in $P_t$ are coloured in green.(b) The figure of $GG_{\le t}^\sigma$. Edges belonging to $B_{l_1}$ are coloured in red while those belonging to $B_{l_2}$ are coloured in blue.(c) The figure of $HH_{\le t}^\sigma$. The introduced $v_l$ is coloured in violet.  (d) The figure of $GH_{\le t}^\sigma$. The gadget $LL_{\le l}^\sigma$ is enclosed by the black ellipse and the path gadgets are coloured in violet. (e) Visualization of the fact that $B_{l_1}\cup B_{l_2}\subseteq B_t\cup P_t\cup E_{X_l}$. } 
	\label{fig:K7alg1}
\end{figure*}
    \begin{lemma}
         Suppose that $t$ is a child of $d$ and $l_1,l_2,...,l_q$ are all of $t$'s children with the navel $X$ in $(T,\beta)$. The boundary mappings $BM_{l_1},...,BM_{l_q}$ of $G_{\le l_1},...,G_{\le l_q}$ are given. 
         For each $\sigma\in \{0,1\}^{|B_t\cup P_t|}$ which is an assignment of edges in $B_t\cup P_t$, $HH_{\le t}^\sigma$ is obtained in the following way: delete all the vertices in $\beta(l_i)-\beta(t)$ for each $1\le i\le q$ in $GG_{\le t}^\sigma$. Add a new vertex $v_l$ and an edge $(v_l,x)$ for each $x\in X$. The signatures assigned to $HH_{\le t}^\sigma-X-\{v_l\}$ keeps the same as those in $GG_{\le t}^\sigma$. For each vertex $x\in X$, suppose a permutable matchgate signature $f$ of arity $k$ is assigned to $x$ in $GG_{\le t}^\sigma$, $S_1=\{s|\text{the }s\text{th edge has another endpoint in }\bigcup_{1\le i\le q}\beta(l_i)-\beta(t)\}$ and $(S_1,S_2)$ is a partition of $\{1,\dots,k\}$. Then in $HH_{\le t}^\sigma$, $x$ is assigned the signature same as that assigned to the tail vertex of the $(S_1,S_2)$-path gadget of $f$.

         A matchgate signature $f_{\le l}$ of arity $|X|$ can be computed in polynomial time, such that if $v_l$ is assigned $f_{\le l}$, then $Z(HH_{\le t}^\sigma)=Z(GG_{\le t}^\sigma)$.
         \label{MGmulrepre}
    \end{lemma}
\begin{proof}
     For each $\sigma\in \{0,1\}^{|B_t\cup P_t|}$ which is an assignment of edges in $B_t\cup P_t$, the signature assigned to each vertex in $GG_{\le t}^\sigma$ is a permutable matchgate signature by Lemma \ref{lemPin}. 
     
     We obtain the graph $GH_{\le t}^\sigma$ by replacing each $x\in X$ in $GG_{\le t}^\sigma$ with the $(S_1,S_2)$-path gadget $PP_{(S_1,S_2)}^x$ introduced in Definition \ref{defpath}. By Lemma \ref{lempath}, $Z(GH_{\le t}^\sigma)=Z(GG_{\le t}^\sigma)$. We denote the head vertex of $PP_{(S_1,S_2)}^x$ as $x_l$ and the tail vertex as $x_t$.

Notice that the signature assigned to each $v\in HH_{\le t}^\sigma-X-\{v_l\}$ in $HH_{\le t}^\sigma$ is same as that assigned to $v$ in $GH_{\le t}^\sigma$.
    Each signature assigned to $x\in X$ in $HH_{\le t}^\sigma$ is same as that assigned to $x_t$ in $GH_{\le t}^\sigma$. Besides, $v_l\in HH_{\le t}^\sigma$ corresponds to a gadget in $GH_{\le t}^\sigma$. We denote the gadget as $LL_{\le l}^\sigma$.
    
    We define $X_l=\bigcup_{x\in X} (PP_{(S_1,S_2)}^x-x_t)\subseteq V(GH_{\le t}^\sigma)$. The gadget $LL_{\le l}^\sigma$ is a subgraph of $GH_{\le t}^\sigma$, induced by vertices in $ \bigcup_{1\le i \le q} GH_{\le l_i}^\sigma \cup X_l$. As $LL_{\le l}^\sigma$ only consists of matchgate signatures, it satisfies the parity condition. The signature $f_{\le l}$ of the gadget $LL_{\le l}^\sigma$ is of arity no greater than 3, and consequently is also a matchgate signature by Lemma \ref{MG<3}. As a result, if the matchgate signature $f_{\le l}$ is assigned to $v_l$ in $HH_{\le t}^\sigma$, then $Z(HH_{\le t}^\sigma)$ would equal $Z(GH_{\le t}^\sigma)$, and $Z(GG_{\le t}^\sigma)$ as well. 

    The final step is to demonstrate how to compute $f_{\le l}$ in polynomial time. We use $E_{X_l}$ to denote all the edges with one endpoint in $X_l$ in $GH_{\le t}^\sigma$. For each assignment $\tau\in \{0,1\}^{|E_{X_l}|}$, the value of the gadget under the assignment $\tau$ can be computed in the following way: for each $y\in X_l$, the value of the signature $f_y$ assigned to $y$ can be computed in $O(1)$ time, since the value of all its variables are determined. For each $1\le i\le q$, $\sigma$ and $\tau$ together form an assignment of the boundary edge set $B_{l_i}$ since each edge in $B_{l_i}$ are incident to a vertex in $\beta(l_i)\cap \beta(t)\subseteq X_l\cup A_t$. Consequently, $Z(GG_{\le l_i}^{\sigma\tau})=BM_{l_i}(\sigma\tau)$ can also be computed in $O(1)$ time since $BM_{l_i}$ is already given. The value of the assignment can be written as $\prod_{y\in X_l}f_y(\sigma\tau)\cdot\prod_{1\le i\le q}BM_{l_i}(\sigma\tau)$ and can be computed in $O(n)$ time as $q\le n$. 
    \par Since $|E_{X_l}|\le 3k+9$,  by Lemma \ref{exhaustsum} and \ref{exhaustgad}, $f_{\le l}$ can also be computed in polynomial time.
\end{proof}
 See Figure \ref{fig:K7alg1} for a visualization of the algorithm. Now we prove Lemma \ref{succK7}. 
 \begin{proof}
        For each $\sigma\in \{0,1\}^{|B_t\cup P_t|}$, we sequentially replace children with the same navel $X_i$ with a vertex $v_i$ assigning a matchgate signature, as described in Lemma \ref{MGmulrepre}. The value of $\omega(\sigma)$ remains unchanged under these transformations.
        \par Remember that $G_t-A_t$ can be embedded on a surface with genus at most $h$. After the replacement, the obtained graph can also be embedded on a surface with genus at most $h$, since the transformations only introduce at most one vertex of arity 3 in each face, and other introduced vertices are of arity 1 or 2. Consequently, $Z(GG_{\le t}^\sigma)$ can be computed in polynomial time by Lemma \ref{mggealg}. Since $|B_t\cup P_t|\le k(h+3)$, $BM_t$ can also be computed in polynomial time by Lemma \ref{exhaustsum} and \ref{exhaustgad}.
    \end{proof}
   With Lemma \ref{succK7}, the algorithm can compute the boundary mappings in the reverse order of BFS and finally output the value of $G$ in polynomial time, and thus Theorem \ref{K7alg} is proved.

\section{Hardness for bounded degree \#CSP on minor-free graphs}\label{secK8}

In this section, we prove the \#P-hardness for maximum degree $k$ \#CSP that forbids a minor set. If $k=2$, then obviously it's computable in polynomial time. Thus, we may assume $k\ge 3$. In the rest of this section, we prove the fifth statement in Theorem \ref{allthm2}. We restate it as follows.
\begin{theorem}     
     Suppose that $\mathcal{G}\cap\mathcal{H}=\emptyset$. Then for any integer $D\ge 3$,  $\#R_D\text{-CSP}(\mathcal{F})[\mathcal{G}]$ is computable in polynomial time if $\mathcal{F}\subseteq \mathscr{A}$ or $\mathcal{F}\subseteq \mathscr{P}$; otherwise, it is $\text{\#P}$-hard.
    \label{CSPK8}
\end{theorem}
By Lemma \ref{RD=bounded}, the following corollary can be obtained from Theorem \ref{CSPK8}.

\begin{corollary}
     Suppose that $\mathcal{G}\cap\mathcal{H}= \emptyset$, then for any integer $k\ge 3$, maximum degree $k$ $\text{\#CSP}(\mathcal{F})[ \mathcal{G} ]$ is computable in polynomial time if $\mathcal{F}\subseteq \mathscr{A}$ or $\mathcal{F}\subseteq \mathscr{P}$; otherwise it is $\text{\#P}$-hard.
     \label{CSPK8coro}
\end{corollary}
The fifth statement in Theorem \ref{allthm} can be obtained from Corollary \ref{CSPK8coro} by the fact that $K_8\notin \mathcal{H}$ \cite{thilikos2022killing}. Each planar graph is a minor of $H_k$ for some $k\in \mathbb{N}^+$, and thus belongs to $\mathcal{H}$. As $\mathcal{G}\cap\mathcal{H}= \emptyset$ and $\mathcal{H}$ is closed under taking minors, each graph in $\mathcal{H}$ is also $\mathcal{G}$-minor-free. Consequently by Remark \ref{remarCSP}, it is sufficient to prove the following lemma.

\begin{lemma}     
    \par If the graph class $\mathcal{G}$ satisfies $\mathcal{G}\cap\mathcal{H}= \emptyset$, and the signature set $\mathcal{F}$ satisfies that $\mathcal{F}\subseteq \mathscr{M}_P$, $\mathcal{F}\nsubseteq \mathscr{A}$, then for any constant $k\ge 3$, maximum degree $k$ $\text{Holant}(\mathcal{F}|\widehat{\mathcal{EQ}})[ \mathcal{G} ]$ is $\text{\#P}$-hard. 
    \label{K8hard}
     \end{lemma}
    

Since $\mathcal{F}\subseteq \mathscr{M}_P,\mathcal{F}\nsubseteq \mathscr{A}$, there exists a $g\in \mathcal{F}$ satisfying $g\in \mathscr{M}_P-\mathscr{A}$. By Lemma \ref{lem:AsymtoSym}, a symmetric $f\in \mathscr{M}-\mathscr{A}$ can be realized by a planar left-side gadget.
    
    Now, we give a proof outline of Lemma \ref{K8hard}. We start from the fact that $\text{Holant}([0,0,1,0])$ is $\text{\#P}$-hard\cite{xia2007computational}. Our proof consists of 3 parts:
    \begin{enumerate}
        \item Reduce $\text{Holant}([0,0,1,0])$ to $\text{Holant}([0,0,1,0])\langle\mathcal{R}\rangle$;
        \item Realize $[0,0,1,0]$ with signatures in $\{f,[1,0],[1,0,1],[1,0,1,0]\}$ in a planar way;
        \item Reduce  $\text{Holant}([0,0,1,0])$ to $\text{Holant}(\mathcal{F}|\widehat{\mathcal{EQ}})[\mathcal{G}]$.
    \end{enumerate}
    The following theorem is our starting point.
\begin{theorem}[\cite{xia2007computational}]
    $\text{Holant}([0,0,1,0])$ is \#P-hard.
    \label{K8origin}
\end{theorem}
   \subsection{Move the crossings}
    This section addresses Part 1. We commence with the definition of the graph class $\mathcal{R}$:
\begin{definition}\cite{thilikos2022killing}
    A \textit{ring blowup} of a planar graph $G=(V,E)$ is a graph $G'=(V',E')$ obtained from a planar embedding of $G$. We denote the set of vertices incident to the outer face in the embedding as $V_f\subseteq V$. $V'=V\cup V_f'$, where $V_f'=\{u_v|v\in V_f\}$, i.e. $V_f'$ is a copy of $V_f$. $E'=E\cup \{(u_v,w)|(v,w)\in E\}\cup\{(u_v,u_w)|(v,w)\in E\}\cup\{(u_v,v)|v\in V_f\}$.
    \par A graph $G$ is called a ring blowup graph if it is a subgraph of a ring blowup of some planar graph. We use $\mathcal{R}$ to denote the class of ring blowup graphs.
\end{definition}
\begin{figure}
            \centering
            \includegraphics[height=0.2\textheight]{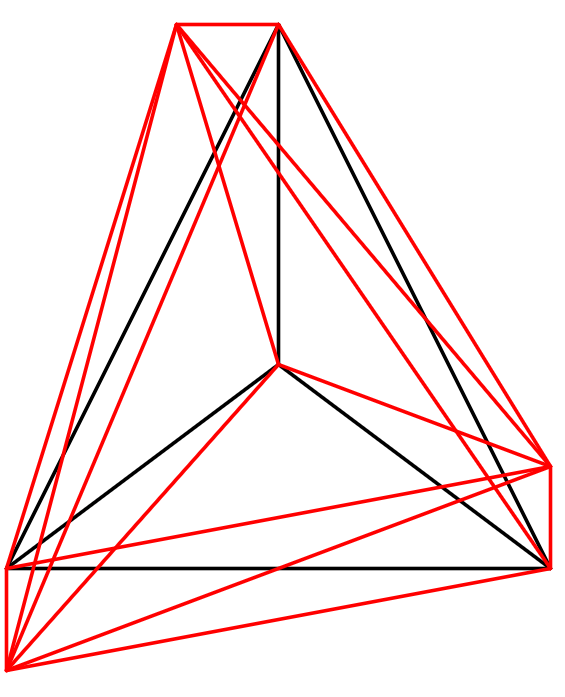}
            \caption{A ring blowup of $K_4$. $K_4$ is coloured in black.}
            \label{fig:ringblow}
        \end{figure}
See Figure \ref{fig:ringblow} for an example.
To begin with, counting perfect matchings on a 3-regular graph, which can also be denoted as $\text{Holant}([0,0,1,0])$, is $\text{\#P}$-hard\cite{xia2007computational}. Using the method in \cite{curticapean2022parameterizing}, we can achieve a similar conclusion:
\begin{lemma}
    $\text{Holant}([0,0,1,0])\le_{T}\text{Holant}([0,0,1,0])\langle \mathcal{R} \rangle$
    \label{movecrossing}
\end{lemma}
\begin{proof}
     Firstly, we prove that $\text{Holant}([0,0,1,0])\le_{T}\text{Holant}([0,0,1,0],[1,0,-1])\langle \mathcal{R} \rangle$. The method is exactly the same as that in \cite[Lemma 3.1]{curticapean2022parameterizing}, and here we only present a sketch of the proof. See \cite[Lemma 3.1]{curticapean2022parameterizing} for more details. 
     
     Give any instance of $\text{Holant}([0,0,1,0])$ with underlying graph $G$, we embedded $G$ on the plane, possibly with crossings. We also draw a circle $C$ on the plane, such that each vertex and each edge in $G$ is enclosed by $C$. By the process in  \cite[Lemma 3.1]{curticapean2022parameterizing}, we can obtain a graph $G_{oc}$, such that $Z(G_{oc})=Z(G)$, each vertex in $G_{oc}$ is still enclosed by $C$, and there is no crossings inside $C$. In contrast, all the crossings in $G_{oc}$ are situated outside $C$ and have the sole form such that each line segment crosses exactly 1 other line segment, as depicted in Figure \ref{fig:outcross}. We call such $G_{oc}$ an \textit{outside crossing graph}. It can be verified that an outside crossing graph is always a ring blowup graph.
\begin{figure*}
	\centering
    \subcaptionbox{\label{fig:outcross}}{ \includegraphics[width=0.3\textwidth]{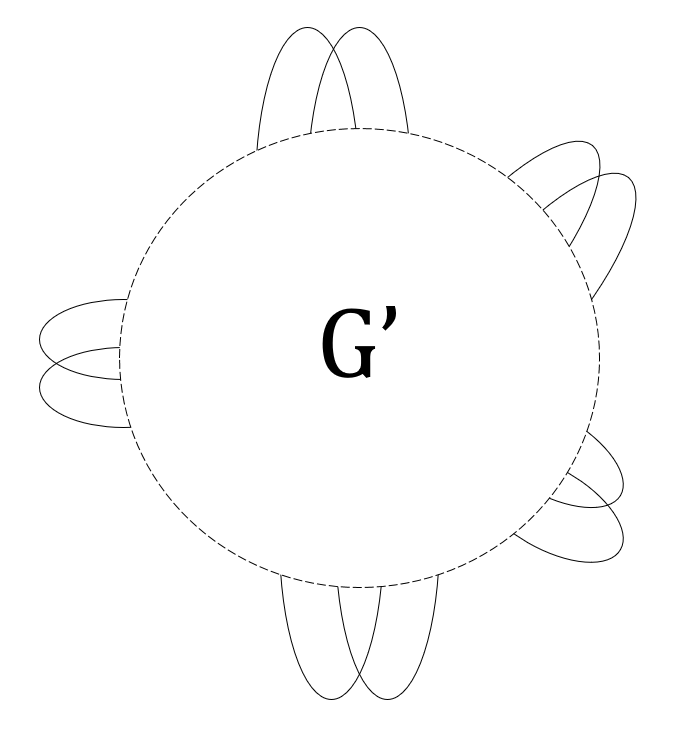}}
		\subcaptionbox{\label{fig:crossgg}}{\includegraphics[width=0.3\textwidth]{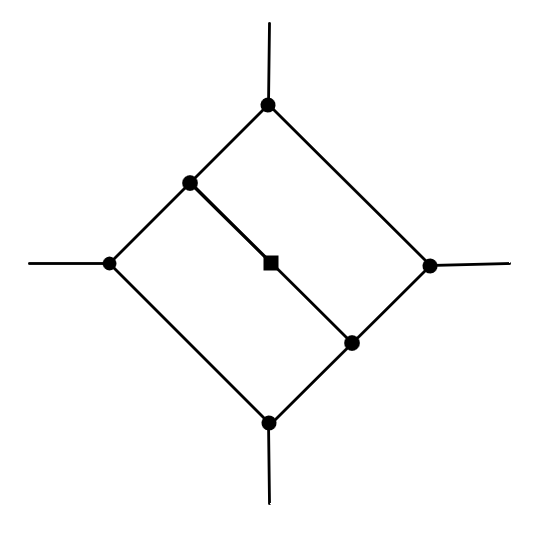}}
	\subcaptionbox{\label{fig:gad[1,0,2]}}{\includegraphics[width=0.3\textwidth]{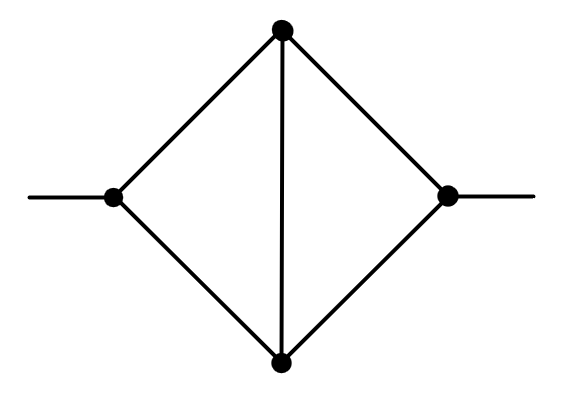}}
	\caption{(a) A visualization of $G_{oc}$ in Lemma \ref{movecrossing}. $G_{oc}$ has no crossing inside the dotted circle. (b) The construction of the crossing gadget. Each vertex of degree 3 represented by a circle is assigned $[0,0,1,0]$, while the vertex in the middle of degree 2 represented by a square is assigned $[1,0,-1]$. (c) The construction of the gadget $GG_1$. Each vertex is assigned $[0,0,1,0]$.} 
	\label{fig:movecross}
\end{figure*}

 For each $v\in G$, the degree of $v$ is equal to that of the corresponding vertex $v'\in G_{oc}$ by the aforementioned process. Besides, each vertex introduced in $G_{oc}$ belongs to some crossing gadget(Figure \ref{fig:crossgg}), which consists of 6 vertices of degree 3 and 1 vertices of degree 2. Thus $G_{oc}$ is also a 3-regular graph with each vertex assigned a $[0,0,1,0]$ signature, except vertices of degree 2 assigned the $[1,0,-1]$ signatures introduced in the crossing gadgets. 
 
 Secondly, we prove $\text{Holant}(\{[0,0,1,0],[1,0,-1]\})\langle \mathcal{R} \rangle\le_{T}\text{Holant}([0,0,1,0])\langle \mathcal{R} \rangle$. Figure \ref{fig:gad[1,0,2]} shows the construction of the planar 3-regular gadget $GG_1$ of the signature $[1,0,2]$. By Corollary \ref{corointerp}, we may interpolate $[1,0,-1]$ using $[1,0,2]$. It can also be verified that $G_{oc}$, as well as graphs introduced in the interpolation step, are still outside crossing graphs. Consequently, Lemma \ref{movecrossing} is proved.
\end{proof}

As proved in \cite{thilikos2022killing}, $\mathcal{R}$ is a subset of $fb(\mathcal{G})$ for any $\mathcal{G}\cap\mathcal{H}= \emptyset$, thus Lemma \ref{movecrossing} also shows that $\text{Holant}([0,0,1,0])\le_{T}\text{Holant}([0,0,1,0])[ \mathcal{G} ], \forall \mathcal{G}\cap\mathcal{H}= \emptyset$.

\subsection{Realizing [0,1,0,0]}
In this section, we want to prove the following lemma:
\begin{lemma}
    If $f\in\mathscr{M}-\mathscr{A}$ and is symmetric, then $[0,0,1,0]$ can be simulated by $\{f\}|\{[1,0],[1,0,1],[1,0,1,0]\}$ in a planar left-side manner.
    \label{lem0010simulate}
\end{lemma}
 \begin{figure*}
	\centering
    \subcaptionbox{\label{fig:k-3pin}}{ \includegraphics[width=0.25\textwidth]{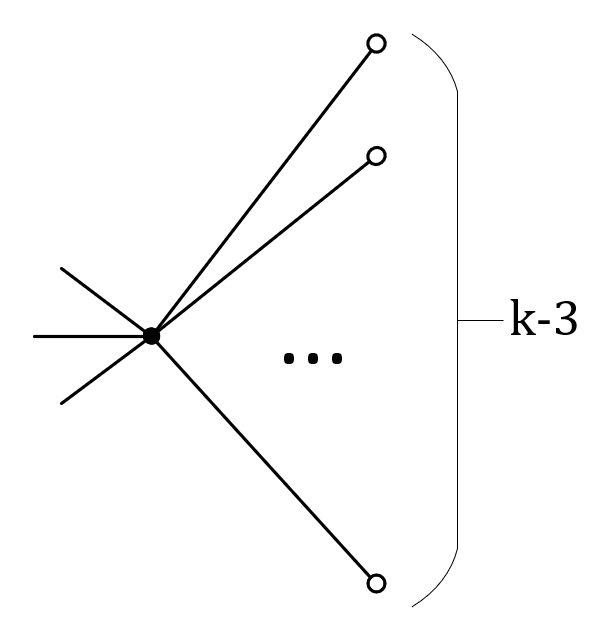}}
		\subcaptionbox{\label{fig:k-2pin}}{\includegraphics[width=0.25\textwidth]{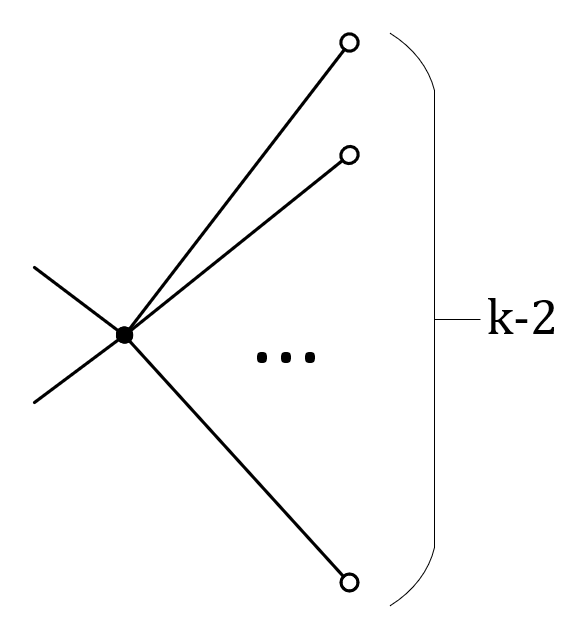}}
	\subcaptionbox{\label{fig:k-3hand}}{\includegraphics[width=0.4\textwidth]{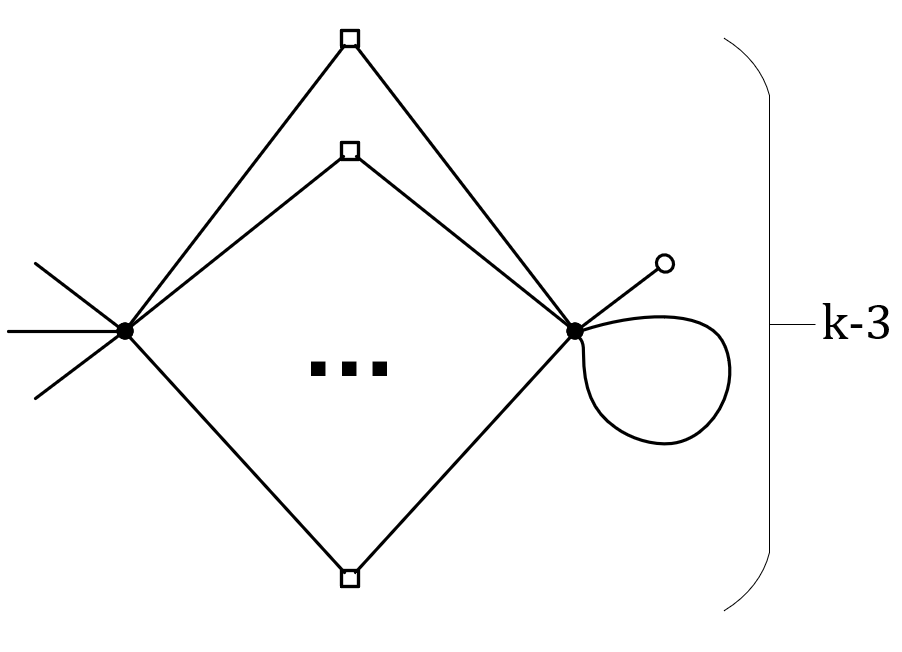}}
    \subcaptionbox{\label{fig:modify3}}{ \includegraphics[width=0.25\textwidth]{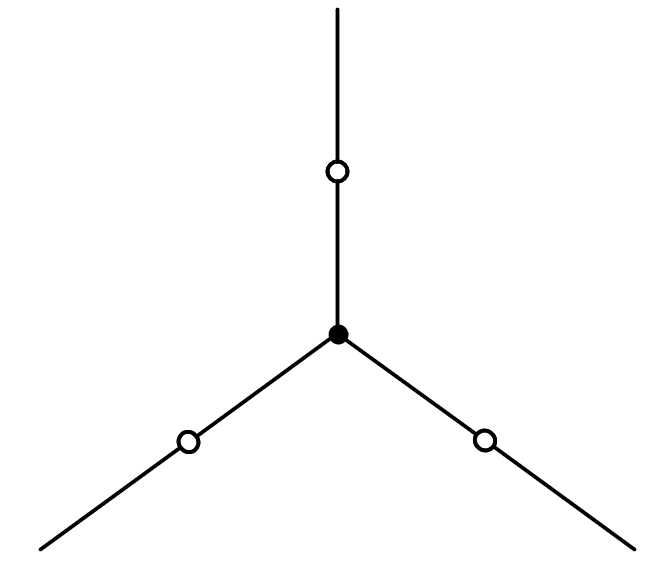}}
		\subcaptionbox{\label{fig:edge4}}{\includegraphics[width=0.35\textwidth]{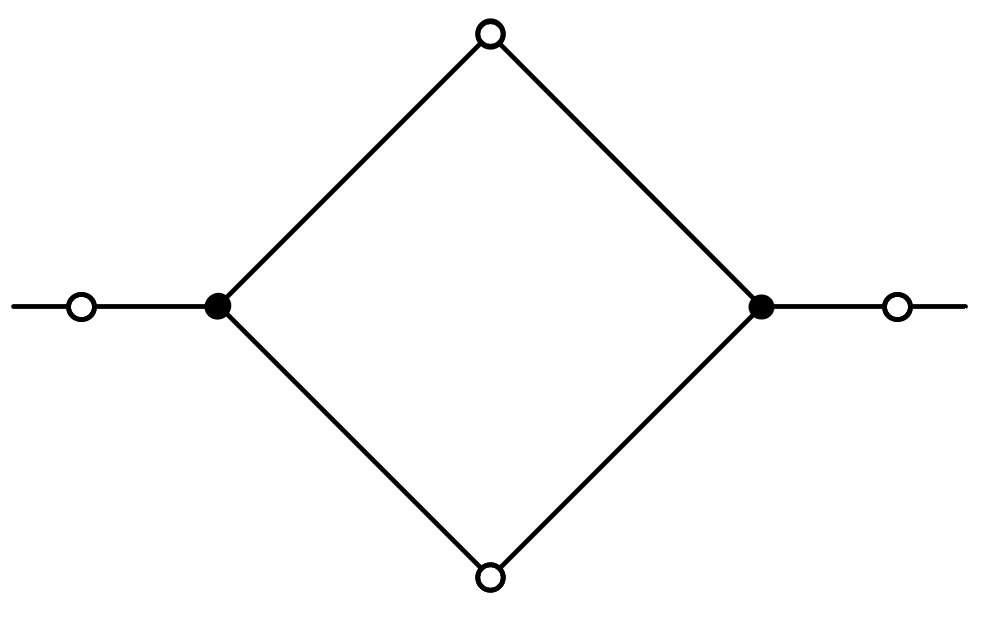}}
	\subcaptionbox{\label{fig:vertex4}}{\includegraphics[width=0.35\textwidth]{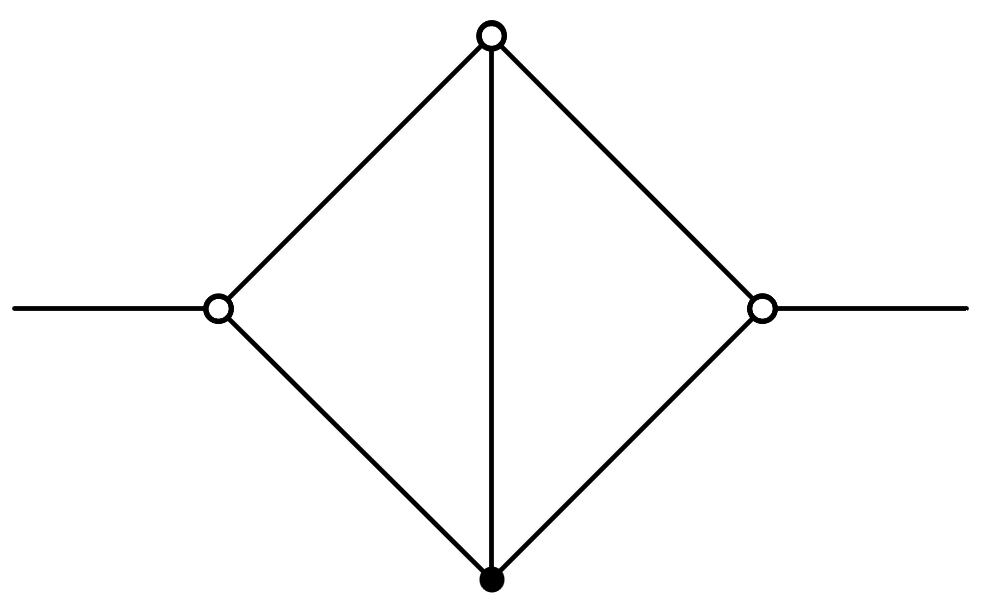}}
	\caption{Gadgets used in the proof of Lemma \ref{lem0010simulate}. The signature assigned to each vertex is context-specific.} 
	\label{fig:gadgets}
\end{figure*}
    It should be noted that the gadget constructions presented here are in accordance with the instructions in Remark \ref{remGadConstruct}.
    
\begin{proof}
    By Lemma \ref{M-Aform}, $f$ have 5 distinct forms. We deal with each form respectively.
    
    \subparagraph{Case 1:} $f=[0,1,0,...,0]_k,k\ge 3$. 
    
    By assigning a $[1,0]$ signature to each vertex of degree 1 labelled by a hollow circle, and $f$ to the vertex of degree $k$ labelled by a solid circle in Figure \ref{fig:k-3pin} and \ref{fig:k-2pin}, we obtain the gadget of signature $[0,1,0,0]$ and $[0,1,0]$ respectively. By assigning a $[0,1,0]$ signature to each vertex of degree 2 labelled by a hollow circle, and a $[0,1,0,0]$ signature to the vertex of degree 3 labelled by a solid circle in Figure \ref{fig:modify3}, we realize the $[0,0,1,0]$ signature.

    \subparagraph{Case 2:} $f=[0,...,0,1,0]_k,k\ge 3$. 
    
    By assigning a $[1,0]$ signature to the vertex of degree 1 labelled by a hollow circle, an $f$ to each vertex of degree $k$ labelled by a solid circle, and $[1,0,1]$ to each vertex of degree $2$ labelled by a hollow square in Figure \ref{fig:k-3hand}, we realize the $[0,0,1,0]$ signature.

    \subparagraph{Case 3:}$f=[1,0,r], r^4\neq 0,1$.
    \subparagraph{Subcase 3.1:} $|r|\neq 0,1$.
    
     By the same method in the proof of Corollary \ref{corointerp}, we can obtain a series of gadgets of signatures $g_1,...,g_k,...$ where $g_k=[1,0,r^k]$ for each $k\ge 1$. For each $k\ge 1$, by assigning a $g_k$ signature to each vertex of degree 2 labelled by a hollow circle, and a $[1,0,1,0]$ signature to the vertex of degree 3 labelled by a solid circle in Figure \ref{fig:modify3}, we realize the signature $h_k=[1,0,r^{2k},0]=\frac{1}{r^{2k}}[r^{2k},0,1,0]$. Using $h_1,...,h_k,...$, we can interpolate the $[0,0,1,0]$ signature.
    \subparagraph{Subcase 3.2:} $|r|=1$ but $ r^4\neq 1$. 
    
    By assigning a $[1,0,r]$ signature to each vertex of degree 2 labelled by a hollow circle, and a $[1,0,1,0]$ signature to the vertex of degree 3 labelled by a solid circle in Figure \ref{fig:edge4}, we realize the  $[r^2+1,0,2r^2]$ signature. Since $|r|=1$ but $ r^4\neq 1$, we have $|2r^2|=2$ and $0<|r^2+1|<2$, and consequently $[r^2+1,0,2r^2]$ is a multiple of $[1,0,r']$ where $|r'|\neq 0,1$. Then we are done by Subcase 3.1.
    \subparagraph{Case 4:} $f=[1,0,r,0,r^2,0,...]_k,k\ge 3, r^2\neq 0,1$.
     \subparagraph{Subcase 4.1:} $r^4\neq 0,1$. 
     
     By assigning a $[1,0]$ signature to each vertex of degree 1 labelled by a hollow circle, and $f$ to the vertex of degree $k$ labelled by a solid circle in Figure \ref{fig:k-2pin}, we obtain the gadget of signature $[1,0,r]$. We are done by Case 3. 
     \subparagraph{Subcase 4.2:} $r=\mathfrak{i}$(or $r=-\mathfrak{i}$). 
     
    By assigning a $[1,0]$ signature to each vertex of degree 1 labelled by a hollow circle, and $f$ to the vertex of degree $k$ labelled by a solid circle in Figure \ref{fig:k-3pin}, we obtain the gadget of signature $[1,0,\mathfrak{i},0]$(or $[1,0,-\mathfrak{i},0]$). By assigning a $[1,0,\mathfrak{i},0]$(or $[1,0,-\mathfrak{i},0]$) signature to each vertex of degree 3 labelled by a hollow circle, and a $[1,0,1,0]$ signature to the vertex of degree $3$ labelled by a solid circle in Figure \ref{fig:vertex4}, we obtain the gadget of signature $[-1-\mathfrak{i},0,-1-3\mathfrak{i}]$(or $[-1+\mathfrak{i},0,-1+3\mathfrak{i}]$). Again we are done by Case 3. 

    \subparagraph{Case 5:} $f=[0,1,0,r,0,r^2,...]_k,k\ge 3, r^2\neq 0,1$

    By assigning a $[1,0]$ signature to each vertex of degree 1 labelled by a hollow circle, and $f$ to the vertex of degree $k$ labelled by a solid circle in Figure \ref{fig:k-3pin} and \ref{fig:k-2pin}, we obtain the gadget of signature $[0,1,0,r]$ and $[0,1,0]$ respectively. By assigning a $[0,1,0]$ signature to each vertex of degree 2 labelled by a hollow circle, and a $[0,1,0,r]$ signature to the vertex of degree 3 labelled by a solid circle in Figure \ref{fig:modify3}, we realize the $[r,0,1,0]=\frac{1}{r}[1,0,\frac{1}{r},0]$ signature. We are done by Case 4.

    All possible cases are analyzed, and in each case only planar gadgets are introduced. Consequently, Lemma \ref{lem0010simulate} is proved.
\end{proof}

\subsection{Reduction after realizing}
\par In every realization in Lemma \ref{lem0010simulate}, $[0,0,1,0]$ is simulated by some planar gadget, either through gadget construction or polynomial interpolation. Thus for any outside crossing graph $G_{oc}$, if we replace each vertex in $G_{oc}$ with a planar gadget that simulates $[0,0,1,0]$ as described in the proof of Lemma \ref{lem0010simulate} and Lemma \ref{lem:AsymtoSym}, and each edge in $G_{oc}$ with a $[1,0,1]$ signature, the obtained graph $G_{oc}'$ is still an outside crossing graph, and consequently a ring blowup graph and we have $G_{oc}'\in fb(\mathcal{G})$.

Furthermore, $G_{oc}''$ can be seen as an instance of $\text{Holant}(\mathcal{F}|\widehat{\mathcal{EQ}})[\mathcal{G}]$ by Remark \ref{remGadConstruct} since $[1,0],[1,0,1],[1,0,1,0]\in \widehat{\mathcal{EQ}_{\le 3}}$ by Lemma \ref{transEQ}. Besides, we only introduce signatures from $\widehat{\mathcal{EQ}_{\le 3}}$, and consequently by Theorem \ref{K8origin} the proof of Lemma \ref{K8hard} is completed.




\section{Hardness for \#CSP on apex graphs}\label{secK6}

In this section, we prove the third statement in Theorem \ref{allthm2}. We restate it as follows.
\begin{theorem}     
  $\text{\#CSP}(\mathcal{F})\langle \mathcal{PLA}\rangle$ is computable in polynomial time if $\mathcal{F}\subseteq \mathscr{A}$ or $\mathcal{F}\subseteq \mathscr{P}$; otherwise it is $\text{\#P}$-hard. 
    \label{CSPK6}
     \end{theorem}
     Suppose that $G$ is an apex graph with an apex vertex $a$. Since $G-a$ is planar, $G-a$ is $\{K_5,K_{3,3}\}$-minor-free and consequently $G$ is $\{K_6,K_{3,3,1}\}$-minor-free. This gives $\mathcal{PLA}\subseteq fb(\{K_6,K_{3,3,1}\})$. The third statement in Theorem \ref{allthm} then can be obtained from Theorem \ref{CSPK6}. Obviously, $\mathcal{PL}\subseteq \mathcal{PLA}$, and by Remark \ref{remarCSP}, it is sufficient to prove the following lemma.


\begin{lemma}
     If $\mathcal{F}\subseteq \mathscr{M}_P$ but $\mathcal{F}\nsubseteq \mathscr{A}$, then $\text{Holant}(\mathcal{F}|\widehat{\mathcal{EQ}})\langle \mathcal{PLA} \rangle$ is $\text{\#P}$-hard.
    \label{thmPLA}
\end{lemma}
   
In the following, we reduce the problem of counting matchings on planar 3-regular multigraphs, which is \#P- hard by \cite{xia2007computational}, to a specific problem on $\mathcal{PLA}$ in Section 6.1. Then, we reduce this problem to a \#CSP problem on $\mathcal{PLA}$ by simulating $[0,0,0,1,0]$ in a planar way in Section 6.2. The following theorem is our starting point.
\begin{theorem}[\cite{xia2007computational}]
    Counting matchings on a planar 3-regular multigraph is \#P-hard.
    \label{K6origin}
\end{theorem}
\subsection{Reducing to a problem on apex graphs}

We first reduce the problem of counting matchings on planar 3-regular graphs to the following form:
\par For a 3-regular planar multigraph $G=(V,E)$, we construct $G_a=(V_a,E_a)$ where $V_a=V\cup \{a\}$ and $E_a=E\cup \{va|v\in V\}$. 
We also assign signatures from $(\{[0,0,0,1,0]\}\cup \widehat{\mathcal{EQ}})$ to vertices in $G_a$ as follows: we assign a $[0,0,0,1,0]$ signature to each vertex in $V$, and a $[1,0,1,0,1,...,1]$ signature to $a$. 

Now we show that the value of $Z(G_a)$ is exactly the number of matchings
 in $G$.
 \begin{lemma}
     The value of $Z(G_a)$, as previously defined, is exactly the number of matchings in $G$.
 \end{lemma}
 \begin{proof}
      For each assignment $\sigma$ of $E_a$, $\omega(\sigma)$ can only be either 0 or 1 since the value of each signature in $\{[0,0,0,1,0]\}\cup \widehat{\mathcal{EQ}}$  can only provide a multiplier with value 0 or 1. Furthermore, if $\omega(\sigma)=1$, then for each $v\in V$, only one edge incident to $v$ can be assigned the value 0.  As a result, all the edges that belong to $E$ and are assigned the value 0 form a matching of $G$.
      
On the other hand, if $M$ is a matching in $G$, we can construct an assignment $\sigma$ of $E_a$ such that $\omega(\sigma)=1$. For each edge $e\in M$, $\sigma$ assigns the value 0 to $e$ in $G_a$. For each vertex $v\in V$ not covered by $M$, $\sigma$ assigns the value 0 to the edge $va$ in $G_a$. All other edges in $G_a$ are assigned the value 1. For each $v\in V$, exactly one edge incident to it is assigned 0. Furthermore, since $M$ covers an even number of vertices, the number of edges incident to $a$ and assigned the value 1 is even. Consequently, $\omega(\sigma)=1$.

It can also be verified that, the previous construction in fact builds a bijective mapping between matchings in $G$ and assignments of $E_a$ whose value is 1. 
Consequently, we have the value of $Z(G_a)$ is exactly the number of matchings in $G$.
\end{proof}
\subsection{Realizing [0,0,0,1,0] with [0,0,1,0]}
In this part, we first prove the following lemma:
\begin{lemma}
   For each signature $f\in \mathscr{M}_P-\mathscr{A}$, each $[0,0,0,1,0]$ signature assigned to a vertex in $G_a$ can be simulated by $\{f,[1,0],[1,0,1],[1,0,1,0]\}$, such that after replacing each vertex in $G$ with the corresponding gadget in $G_a$, the obtained graph $G_a'$ is an apex graph with the apex vertex $a$.
   \label{lem3to4simulate}
\end{lemma}
    In fact, $[0,0,0,1,0]$ may not be simulated directly by signatures in $\{f,[1,0],[1,0,1],[1,0,1,0]\}$, since it is of odd parity and all the signatures from $\{f,[1,0],[1,0,1],[1,0,1,0]\}$ can be of even parity. However, since a 3-regular graph has an even number of vertices, the number of vertices assigned the $[0,0,0,1,0]$ signature in $G_a$ is also even,  and the following definition and lemma from \cite{guo2020complexity} would help us overcome this difficulty: 
   
    \begin{definition}\cite[Definition 5.3]{guo2020complexity}
        A planar pairing $M$ in a planar graph $G=(V,E)$ is a set of pairs of vertices, such that $G'=(V,E\cup M)$ is still planar and $M$ is a perfect matching in $G'$. 
    \end{definition}
    \begin{lemma}{\cite[Lemma 5.4]{guo2020complexity}}
    Any 3-regular planar multigraph has a planar pairing, and it can be found in polynomial time.
        \label{lemplanar_pairing}
    \end{lemma}
    Now we are ready to prove Lemma \ref{lem3to4simulate}:
    \begin{figure*}
	\centering
    \subcaptionbox{\label{fig:gad[0,0,0,1,0]}}{ \includegraphics[width=0.4\textwidth]{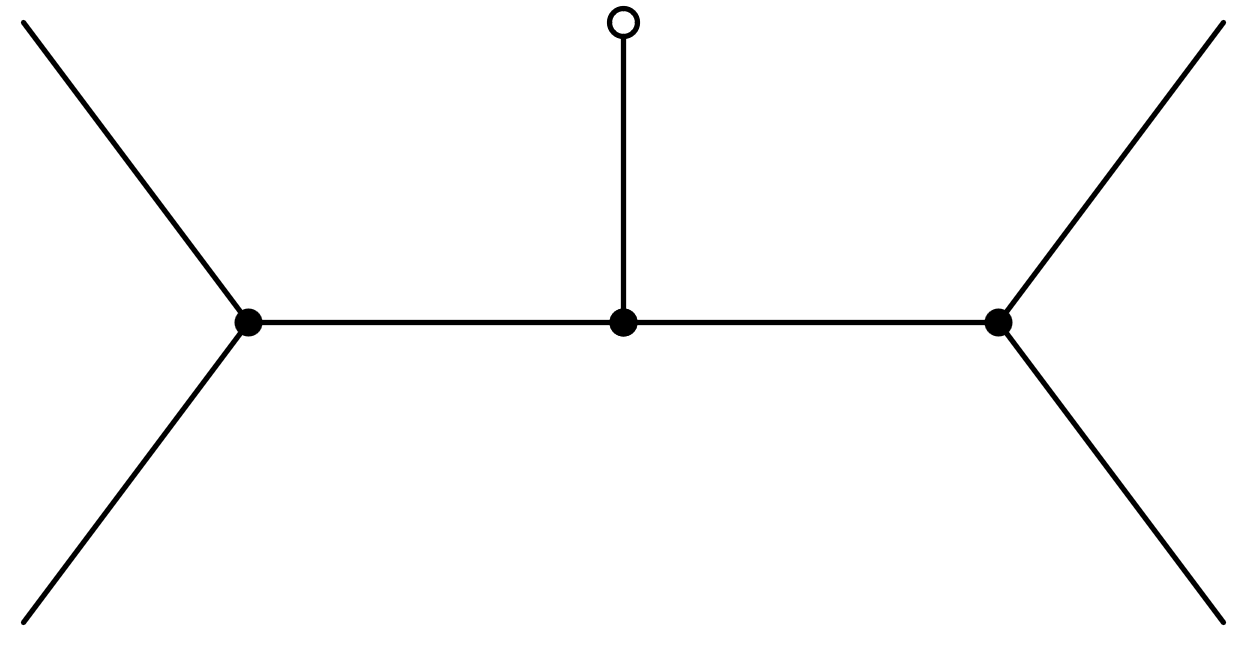}}
		\subcaptionbox{\label{fig:gad[0,0,1]}}{\includegraphics[width=0.2\textwidth]{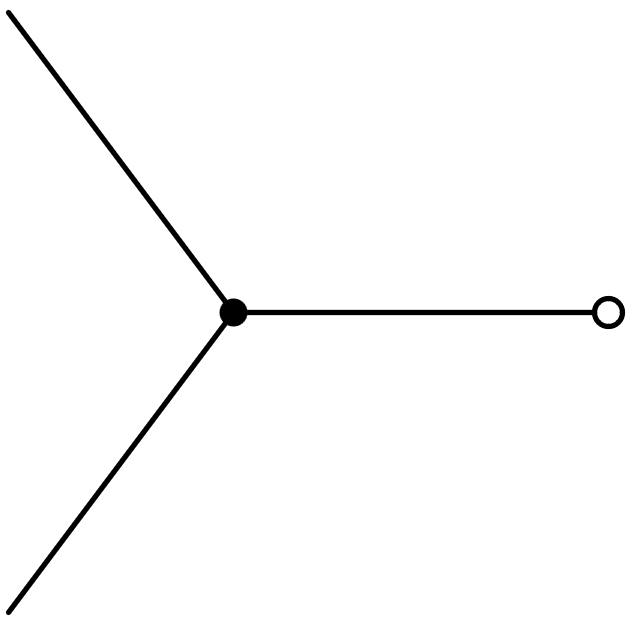}}
	\caption{Gadgets used in the proof of Lemma \ref{lem3to4simulate}.} 
	\label{fig:gadgets6}
\end{figure*}
\begin{proof}
   By Lemma \ref{lem0010simulate}, the signature $[0,0,1,0]$ can be simulated by signatures in $\{f,[1,0],$ $[1,0,1],[1,0,1,0]\}$ in a planar left-side manner. By assigning a $[0,1]$ signature to the vertex of degree 1 labelled by a hollow circle, and $[0,0,1,0]$ to each vertex of degree 3 labelled by a solid circle in Figure \ref{fig:gad[0,0,0,1,0]}, we realize the signature $[0,0,0,1,0]$\footnote{Realizations in this proof also follows the instructions in Remark \ref{remGadConstruct}.}. Now we only need to realize a $[0,1]$ signature for each $[0,0,0,1,0]$. By assigning a $[1,0]$ signature to the vertex of degree 1 labelled by a hollow circle, and $[0,0,1,0]$ to the vertex of degree 3 labelled by a solid circle in Figure \ref{fig:gad[0,0,1]}, we realize the signature $[0,0,1]$, which equals $[0,1]^{\otimes 2}$.  Since $G_a-a$ is a 3-regular planar multigraph, $G_a-a$ has a planar pairing $M$ by Lemma \ref{lemplanar_pairing}.  We can use gadgets of the $[0,0,1]$ signature to replace all the $[0,1]$ signatures that appears in the gadgets of the $[0,0,0,1,0]$ signature, based on the planar pairing $M$. 
   Notice that for each vertex $v\in G_a-a$, since the $[0,0,0,1,0]$ signature assigned to it is symmetric, when replacing $v$ with the gadget of $[0,0,0,1,0]$ in Figure \ref{fig:gad[0,0,0,1,0]}, the $[0,1]$ signature may appear in any face incident to $v$ in $G_a-a$ by rotating the gadget properly. Thus after replacing all the $[0,1]$ signatures with gadgets of $[0,0,1]$ signature, the obtained graph, which we denoted as $G_a'$, is still an apex graph with the same apex vertex $a$.
\end{proof}

Again, by Remark \ref{remGadConstruct}, $G_a'$ is an instance of $\text{Holant}(\mathcal{F}|\widehat{\mathcal{EQ}})\langle \mathcal{PLA} \rangle$ since $[1,0], [1,0,1],$ $[1,0,1,0], [1,0,1,0,1,...,1] \in \widehat{\mathcal{EQ}}$ by Lemma \ref{transEQ}, and consequently by Theorem \ref{K6origin} the proof of Lemma \ref{thmPLA} is completed.

Combining Theorems \ref{CSPK4}, \ref{CSPK5}, \ref{CSPK6}, \ref{CSPK7} and \ref{CSPK8}, our main result Theorem \ref{allthm2}, along with Theorem \ref{allthm}, has been proved.

\section{Conclusions and future directions}\label{conclusion}
In this article, we prove dichotomies for \#CSP and bounded degree \#CSP on a selection of typical and pivotal graph classes, and thus obtain the complete \#CSP and bounded degree \#CSP dichotomies on graphs that forbids an arbitrary clique as a minor.

We remark that several minor-closed graph classes have been introduced in this article, such as $\mathcal{PLA}$ and  $\mathcal{TWG}_h$. Identifying the corresponding forbidden minor set for these classes would be of independent interest.

It is worth study to give a dichotomy for \#CSP and bounded degree \#CSP on arbitrary minor-closed graph class. In particular, we wonder whether there exists any minor-closed graph class $\mathcal{C}$ satisfying $\mathcal{PL}\nsubseteq \mathcal{C}$, such that new nontrivial tractable cases other than those in \cite{cai2017holographicuni} can be found in \#CSP or bounded degree \#CSP on $\mathcal{C}$.

Extending our results to sym-Holant is also worth study. For sym-Holant, there exists a tractable case that is not able to be subsumed by FKT algorithm on planar graphs \cite{cai2022fkt}. This may present certain challenges in the study. Besides, an algorithm for sym-Holant on graphs with bounded treewidth is of independent interest and will be involved in another paper of ours. 





\bibliography{ref}

\appendix

\end{document}